\documentclass[a4paper,11pt,envcountsect]{llncs}
\usepackage{latexsym,amssymb,amsmath,amsfonts,ifthen,owna4}

\newtheorem{postulate}{Postulate}
\renewcommand{\labelenumi}{(\roman{enumi})}
\spnewtheorem{fact}{Fact}{\bfseries}{\itshape}
\spnewtheorem*{remark*}{Remark}{\bfseries}{\rmfamily}

\title{A Behavioural Theory of Recursive Algorithms}

\author{Egon B\"orger\inst{1}, Klaus-Dieter Schewe\inst{2}}

\institute{Universit\`{a} di Pisa, Dipartimento di Informatica, Pisa, Italy,
\email{boerger@di.unipi.it}
\and 
Zhejiang University, UIUC Institute, Haining, China,
\email{kdschewe@acm.org}}

\makeatletter
\def~{\ifmmode\;\else\penalty\@M\ \fi}
\def\@setmcodes#1#2#3{{\count0=#1 \count1=#3
  \loop \global\mathcode\count0=\count1 \ifnum \count0<#2
  \advance\count0 by1 \advance\count1 by1 \repeat}}
\DeclareSymbolFont{italic}{OT1}{\rmdefault}{m}{it}
\let\mathit\undefined
\DeclareSymbolFontAlphabet{\mathit}{italic}
\edef\@tempa{\hexnumber@\symitalic}
\@setmcodes{`A}{`Z}{"7\@tempa41}
\@setmcodes{`a}{`z}{"7\@tempa61}
\makeatother
\newdimen\asmindent     
\asmindent=\parindent
\newcount\asmi
\def\inc{\global\advance\asmi by 1}
\def\dec{\global\advance\asmi by-1}
\def\nl{{}$\par\hangindent\asmi em
  \noindent\kern\asmi em\ignorespaces$} 
\def\asmskip{{}$\par\smallskip\hangindent\asmi em
  \noindent\kern\asmi em\ignorespaces$} 
\def\asm{\global\asmi=0 
 \def\+{\inc\nl}
 \def\-{\dec\nl}
 \def\\{\nl}
 \begin{trivlist}\item[]\leftskip=\asmindent\relax$}
\def\endasm{$\end{trivlist}}
\def\asmarray{\begin{array}[t]{@{}l@{\;}l@{\;}l@{}}}
\def\endasmarray{\end{array}}
\newcount\asmii
\def\subasm{\vtop\bgroup\asmii=0\normalbaselines
 \def\nl##1{$\egroup\advance\asmii by##1\relax\hbox\bgroup\hskip\asmii em$}
 \def\\{\nl{0}}
 \def\+{\nl{1}}
 \def\-{\nl{-1}}
 \hbox\bgroup\hskip\asmii em$}
\def\endsubasm{$\egroup\egroup}
\def\AMB      {\mathrel{\mathbf{amb}}}
\def\ASM#1{\hbox{\sc#1}}        

\def\AND     {\mathrel{\mathbf{and}}}

\def\CHOOSE  {\mathrel{\mathbf{choose}}}

\def\NEW  {\mathrel{\mathbf{new}}}
\def\DO      {\mathrel{\mathbf{do}}}
\def\ELSE    {\mathrel{\mathbf{else}}}

\def\FORALL  {\mathrel{\mathbf{forall}}}
\def\FORSOME  {\mathrel{\mathbf{forsome}}}

\def\IF      {\mathrel{\mathbf{if}}}
\def\IFF      {\mathrel{\mathbf{iff}}}
\def\IMPORT  {\mathrel{\mathbf{import}}}
\def\IN      {\mathrel{\mathbf{in}}}
\def\LET     {\mathrel{\mathbf{let}}}
\def\SELF    {\mathrel{\mathbf{self}}}

\def\NOT     {\mathrel{\mathbf{not}}}

\def\PAR     {\mathrel{\mathbf{par}}}

\def\THEN    {\mathrel{\mathbf{then}}}
\def\WHERE   {\mathrel{\mathbf{where}}}

\def\WITH    {\mathrel{\mathbf{with}}}

\makeatletter
\def\enumerate{%
  \ifnum \@enumdepth >\thr@@\@toodeep\else
    \advance\@enumdepth\@ne
    \edef\@enumctr{enum\romannumeral\the\@enumdepth}%
      \expandafter
      \list
        \csname label\@enumctr\endcsname
        {\usecounter\@enumctr\def\makelabel##1{\hss\llap{##1}}
         \itemsep 0pt\parskip 0pt\parsep 0pt\topsep\smallskipamount}%
  \fi}
\def\itemize{%
  \ifnum \@itemdepth >\thr@@\@toodeep\else
    \advance\@itemdepth\@ne
    \edef\@itemitem{labelitem\romannumeral\the\@itemdepth}%
    \expandafter
    \list
      \csname\@itemitem\endcsname
      {\def\makelabel##1{\hss\llap{##1}}
       \itemsep 0pt\parskip 0pt\parsep 0pt\topsep\smallskipamount}%
  \fi}
\makeatother
\begin{document}

\maketitle

\begin{abstract}

``What is an algorithm?'' is a fundamental question of computer science. Gurevich's behavioural theory of sequential algorithms (aka the sequential ASM thesis) gives a partial answer by defining (non-deterministic) sequential algorithms axiomatically, without referring to a particular machine model or programming language, and showing that they are {\em captured} by (non-deterministic) sequential Abstract State Machines (nd-seq ASMs). Moschovakis pointed out that recursive algorithms such as \textit{mergesort\/} are not covered by this theory. In this article we propose an axiomatic definition of the notion of \emph{sequential recursive algorithm} which extends Gurevich's axioms for sequential algorithms by a Recursion Postulate and allows us to prove that sequential recursive algorithms are captured by \emph{recursive Abstract State Machines}, an extension of nd-seq ASMs by a CALL rule. Applying this recursive ASM thesis yields a characterization of sequential recursive algorithms as finitely composed concurrent algorithms all of whose concurrent runs are partial-order runs.

\end{abstract}

\section{Introduction}

The notion of an algorithm is fundamental for computing, so it may seem surprising that there is still no commonly accepted definition. This is different for the notion of computable function that is captured by several equivalent formalisms such as Turing machines, random access machines, partial-recursive functions, $\lambda$-definable functions and many more \cite{boerger:1989}. However, as there is typically a huge gap between the abstraction level of an algorithm and the one of Turing machines, Gurevich concluded that the latter ones cannot serve as a definition for the notion of an algorithm \cite{gurevich:ams1985}. He  proposed to extend Turing's thesis to a new thesis, based on the observation that ``if an abstraction level is fixed (disregarding low-level details and a possible higher-level picture) and the states of an algorithm reflect all the relevant information, then a particular small instruction set suffices to model any algorithm, never mind how abstract, by a generalised machine very closely and faithfully''.

Still it took many years from the formulation of this new thesis to the publication of the behavioural theory of sequential algorithms in \cite{gurevich:tocl2000}. In this seminal work---also known as the ``sequential ASM thesis''---a \emph{sequential algorithm} (seq-algorithm) is defined by three postulates\footnote{A mathematically precise formulation of these postulates requires more care, see below, but the rough summary here will be sufficient for now.}:

\begin{description}

\item[Sequential Time.] A sequential algorithm proceeds in sequential time using states, initial states and transitions from states to successor states.

\item[Abstract State.] States are universal algebras (aka Tarski structures), i.e. functions resulting from the interpretation of a signature, i.e. a set of function symbols, over a base set.

\item[Bounded Exploration.] There exists a finite set of ground terms such that the difference between a state and its successor state is uniquely determined by the values of these terms in the state\footnote{This set of terms is usually called a {\em bounded exploration witness}, while the difference between a state and its successor is formally given by an {\em update set}. Informally, bounded exploration requires that there are only finitely many terms, the interpretation of which determine how a state will be updated by the algorithm to produce the successor state.}.

\end{description}

The behavioural theory further comprises the definition of  {\em sequential Abstract State Machines} (seq-ASMs) and the proof that seq-ASMs {\em capture} seq-algorithms, i.e. they satisfy the postulates, and every seq-algorithm can be step-by-step simulated by a seq-ASM. As pointed out in \cite[Sect.9.2]{gurevich:tocl2000} (and elaborated to a full proof in \cite[Sect.7.2.2 -7.2.3]{boerger:2003}) it is easy to extend the theory to cover also bounded non-determinism, using non-deterministic sequential ASMs (nd-seq ASMs)\footnote{It suffices to slightly modify the sequential time and the abstract state postulates, using in particular a successor relation instead of a function and permitting the choice between finitely many rules. Gurevich uses the term `bounded-choice nondeterministic algorithm' instead of `nd-seq algorithm'.}.

It should be noted that the {\em definition} of a sequential algorithm given by Gurevich does not require a particular formalism for the {\em specification}. Seq-ASMs capture seq-algorithms, so they are a suitable candidate for specification\footnote{In particular, as pointed out in \cite{boerger:2003}, rules in an ASM look very much like pseudo-code, so the appearance of ASM specifications is often close to the style, in which algorithms have been described in the past and in textbooks. The difference is of course that the semantics of ASMs is precisely defined.}, but they are not the only possible choice. For instance, in the light of the proofs in \cite{gurevich:tocl2000} it is not an overly difficult exercise to show that deterministic Event-B \cite{abrial:2010} or B \cite{abrial:2005} also capture seq-algorithms. 

We believe that in order to obtain a commonly acceptable definition of the notion of algorithm, this distinction between an {\em axiomatic definition} (as by Gurevich's postulates for seq-algorithms), which does not refer to a particular language or programming style, and the {\em capture} by an abstract machine model (such as seq-ASMs, deterministic Event-B or others) is fundamental.

In \cite{moschovakis:2001} Moschovakis raised the question how recursive algorithms like the well-known \textit{mergesort\/} are covered. He questions that algorithms can be adequately defined by machines (see also \cite{Moschovakis19}). Although the perception that Gurevich used seq-ASMs as a definition for the general notion of algorithm---not only for its sequential instance---is a misinterpretation, unfortunately the response by Blass and 
Gurevich  \cite{blass:beatcs2002} to Moschovakis's criticism does not clarify the issue in a convincing way. Instead of admitting that an extended behavioural theory for recursive algorithms still needs to be developed, distributed ASMs with a semantics defined through partial-order runs \cite{gurevich:lipari1995} are claimed to be sufficient to capture recursive algorithms.\footnote{The definition of recursive ASMs in \cite{gurevich:jucs1997} uses a special case of this translation of recursive into distributed computations.} 
As B\"orger and Bolognesi point out in their contribution to the debate \cite{boerger:asm2003}, a much simpler extension of seq-ASMs suffices for the specification of algorithms of the \textit{mergesort\/} kind, which define recursive functions (as do Moschovakis' systems of recursive equations called `recursive programs' \cite{Moschovakis19}, see the discussion in Sect.\ref{sec:more}). 
Furthermore, the response by Blass and 
Gurevich blurs the subtle distinction between the axiomatic definition and the possibility to express any algorithm on an arbitrary level of abstraction by an abstract machine. This led also to Vardi's almost cynical comment that the debate is merely about the preferred specification style (functional or imperative), which is as old as the field of programming \cite{vardi:casm2012}.\footnote{This debate, however, is still much younger than the use of the notion of algorithm.}

While the difficult epistemological issue concerning the definition of the general notion of algorithm has been convincingly addressed for \emph{sequential} algorithms by Gurevich's behavioural theory, no such theory for recursive algorithms or distributed algorithms was available at the time of the debate between Moschovakis, Blass and Gurevich, B\"orger and Bolognesi, and Vardi. In the meantime a behavioural theory for concurrent algorithms has been developed \cite{boerger:ai2016}. It comprises an axiomatic definition of the notion of concurrent algorithm as a family of nd-seq algorithms indexed by agents that is subject to an additional concurrency postulate for their runs, by means of which Lamport's sequential consistency requirement is covered and generalised \cite{lamport:tc1979}. In a nutshell, the concurrency postulate requires that a successor state of the global state of the concurrent algorithm results from simultaneously applying update sets of finitely many agents that have been built on some previous (not necessarily the latest) states.

Using this theory of concurrency it is possible to reformulate the answer given by Blass and Gurevich to Moschovakis's question: every recursive  algorithm is a concurrent algorithm with partial-order runs. Since concurrent ASMs capture concurrent algorithms (as shown in \cite{boerger:ai2016}), they provide a natural candidate for the specification of all concurrent algorithms, thus in particular of recursive algorithms. However, the ``overkill'' argument will remain, as the class of concurrent algorithms is much larger than the class of recursive algorithms.

For example, take the \textit{mergesort\/} algorithm (see Sect.\ref{sect:mergesort}). Every call to (a copy of) itself and every call to (a copy of) the auxiliary \textit{merge\/} algorithm could give rise to a new agent. However, these agents only interact by passing input parameters and return values, but otherwise operate on disjoint sets of locations. In addition, a calling agent always waits to receive return values, which implies that only one or (in case of parallel calls) two agents are active in any state. In contrast, in a concurrent algorithm all agents may be active, and they can interact in many different ways on shared locations as well as on different clocks. As a consequence, concurrent runs may become highly non-deterministic and not linearisable, whereas a sequential\footnote{We use here the attribute `sequential' to emphasize that we view recursive algorithms as sequential algorithms which call sequential algorithms, so that no unbounded parallelism is allowed. The reason is that unbounded parallelism permits to define recursion, as we explain in Sect.\ref{sec:bsp}.} recursive algorithm permits at most bounded non-determinism and without loss of generality several simultaneous calls can always be sequentialised.

This motivates the research we report in this article. Our objective is to develop a behavioural theory of sequential recursive algorithms.
For this we propose an axiomatic definition of sequential recursive algorithms which enriches sequential algorithms by call steps, such that the parent-child relationship between caller and callee defines well-defined shared locations representing input and return parameters. We will present and motivate our axiomatisation in Section \ref{sec:postulates}. 
In Section \ref{sec:recursion} we define recursive ASMs by an appropriate extension of nd-seq ASMs\footnote{Since by definition \emph{recursive} ASMs are extensions of nd-seq ASMs, to obtain a short name we skip the two attributes `non-deterministic' and `sequential'.} with a call rule and show our main result, aka Recursive ASM Thesis:  

{\bf Main Theorem.} 
\emph{Sequential recursive algorithms are captured by recursive ASMs.} 

Section \ref{sec:bsp} is dedicated to an illustration of our theory by examples. We concentrate on \textit{mergesort\/}, \textit{quicksort\/} and the \textit{sieve of Eratosthenes\/} for which we present recursive ASMs. We use the examples to show that the parallelism of ASMs, which is  unbounded, is stronger than sequential recursion, so that there is no need to investigate parallel recursive algorithms separately from parallel algorithms. 

In Section \ref{sec:poruns} we report an application of the recursive ASM thesis.\footnote{A preliminary version of the result presented in Section \ref{sec:poruns} appeared in \cite{BoeSch20a}.} We return to the observation by Blass and Gurevich---though not explicitly stated in \cite{blass:beatcs2002}---that sequential recursive algorithms are linked to concurrent algorithms with partial-order runs. We first show that indeed the runs 
of a sequential recursive algorithm (read: of a recursive ASM) are definable by partial-order runs (Theorem \ref{thm-porun}), which comes at no surprise. 
The amazing second discovery was that also a converse relation holds, namely if all runs of a finitely composed concurrent algorithm (read: of a concurrent ASM $\mathcal{C}$ which consists only of instances of a bounded number of nd-seq ASMs) are definable by partial-order runs, then this algorithm is 
equivalent to a recursive ASM (Theorem \ref{thm-porun'}). This 
relativizes the overkill argument.\footnote{In fact, it shows  that, roughly speaking, finitely composed concurrent algorithms with partial-order runs are indeed the sequential recursive algorithms, and the response given in \cite{blass:beatcs2002} may be seen as the result of ingenious serendipity. However, arbitrary concurrent algorithms as discussed in \cite{boerger:ai2016} are a much wider class of algorithms.} 

Theorem \ref{thm-porun'} can be strengthened if the given concurrent ASM is static, i.e. with a fixed set of agents with associated programs. Such concurrent ASMs are equivalent to nd-seq ASMs (Theorem \ref{thm-Petri}). An interesting corollary of this theorem concerns the Process Rewrite Systems investigated in \cite{Mayr99}. They form the most general and most expressive set in a hierarchy of classes of rewrite systems which can be used to model certain state-based concurrent systems and are classified in \cite{Mayr99} by their expressiveness. Furthermore, the Process Rewrite Systems are shown in \cite{Mayr99} to strictly extend Petri nets (by subroutines that can return a value to their caller), but still to have a decidable reachability problem. As a corollary of Theorem \ref{thm-Petri} it turns out that for each Process Rewrite System, its partial-order runs can be simulated by runs of a nd-seq ASM (Corollary \ref{corollary}). 

Finally, in Section \ref{sec:more} we embed our work into a larger picture of related work on behavioural theories, and in Section \ref{sec:schluss} we present a brief summary and outlook on further research.

\section{Axiomatisation of Recursive Algorithms}\label{sec:postulates}

A decisive feature of a recursive algorithm is that it calls itself, or more precisely a copy (we also say an instance) of itself. If we consider mutual recursion, then this becomes slightly more general, as there is a finite family of algorithms calling (copies of) each other. Therefore, providing copies of algorithms and enabling calls will be essential for the intended definition of the notion of recursive algorithm, whereas otherwise we rely on Gurevich's axiomatic definition of sequential algorithms. Furthermore, there may be several simultaneous calls, which give rise to non-determinism,\footnote{The presence of this non-determinism in recursive algorithms has also been observed in Moschovakis' criticism \cite{moschovakis:2001}, e.g. \textit{mergesort\/} calls two copies of itself, each sorting one half of the list of given elements.} as these simultaneously called copies may run sequentially in one order or the other, or in parallel or even asynchronously. However, there is no interaction between simultaneously called algorithms, which implies that the mentioned execution latitude already covers all choices.

\subsection{Non-deterministic Sequential Algorithms}

In this section we recall the axiomatic definition of non-deterministic sequential algorithms.

\begin{definition}\label{def-nd-alg}\rm

A {\em non-deterministic sequential algorithm} (for short: {\em nd-seq algorithm}) is defined by the branching time, abstract state and bounded exploration postulates \ref{p-time}, \ref{p-state} and \ref{p-bound} in  \cite{gurevich:tocl2000}, paper to which we refer for the motivation for these axioms. 

\end{definition}

\begin{postulate}[Branching Time Postulate]\label{p-time}\rm

An nd-seq algorithm $\mathcal{A}$ comprises a set $\mathcal{S}$, elements of which are called {\em states}, a subset $\mathcal{I} \subseteq \mathcal{S}$, elements of which are called {\em initial states}, and a {\em one-step transition relation} $\tau \subseteq \mathcal{S} \times \mathcal{S}$.\footnote{For deterministic algorithms $\tau$ is a function.} Whenever $\tau(S,S^\prime)$ holds, the state $S^\prime$ is called a {\em successor state} of the state $S$ and we say that the algorithm performs a step in $S$ to yield $S^\prime$.

\end{postulate}

Though Postulate \ref{p-time} only gives a necessary condition for nd-seq algorithms and in particular leaves open what states are, one can already derive some consequences from it such as the notions of {\em run}, {\em final state} and {\em behavioural equivalence}.

\begin{definition}\label{def-seqRun}\rm 

Let $\mathcal{A}$ be a nd-seq algorithm with states $\mathcal{S}$, intial states $\mathcal{I}$ and transition relation $\tau$. A {\em run} of $\mathcal{A}$ is a sequence $S_0, S_1, S_2, \dots$ with $S_i \in \mathcal{S}$ for all $i$ and $S_0 \in \mathcal{I}$ such that $\tau(S_i,S_{i+1})$ holds for all $i$.

\end{definition}

Often $S_i$ is called a {\em final state} of a run $S_0, S_1, S_2, \dots$ of $\mathcal{A}$ (and the run is called {\em terminated} in this state) if $S_j =S_i$ holds for all 
$j \ge i$. But sometimes it is more convenient to use a dynamic termination predicate whose negation guards the execution of the algorithm $\mathcal{A}$ and which is set to true by $\mathcal{A}$ when $\mathcal{A}$ reaches a state one wants to consider as final.

States are postulated to be {\em universal algebras} (aka {\em Tarski structures}), which capture all desirable structures that appear in mathematics and computing.

\begin{definition}\rm

A {\em signature} $\Sigma$ is a finite set of function symbols, and each $f \in \Sigma$ is associated with an {\em arity} $\text{ar}(f) \in \mathbb{N}$. A {\em structure} over $\Sigma$ comprises a {\em base set}\footnote{For convenience to capture partial functions it is tacitly assumed that base sets contain a constant \textit{undef\/} and that each isomorphism $\sigma$ maps \textit{undef\/} to itself.} $B$ and an {\em interpretation} of the function symbols $f \in \Sigma$ by functions $f_B : B^{\text{ar}(f)} \rightarrow B$. An {\em isomorphism} $\sigma$ between two structures is given by a bijective mapping $\sigma : B \rightarrow B^\prime$ between the base sets that is extended to the functions by $\sigma(f_B)(\sigma(a_1),\dots,\sigma(a_n)) = \sigma(f_B(a_1,\dots,a_n))$ for all $a_i \in B$ and $n = \text{ar}(f)$.

\end{definition}

\begin{postulate}[Abstract State Postulate]\label{p-state}\rm

Each nd-seq algorithm $\mathcal{A}$ comprises a {\em signature} $\Sigma$ such that

\begin{enumerate}

\item Each state $S \in \mathcal{S}$ of $\mathcal{A}$ is a structure over $\Sigma$.

\item The sets $\mathcal{S}$ and $\mathcal{I}$ of states and initial states, respectively, are both closed under isomorphisms.

\item Whenever $\tau(S,S^\prime)$ holds, then the states $S$ and $S^\prime$ have the same base set.

\item Whenever $\tau(S,S^\prime)$ holds and $\sigma$ is an isomorphism defined on $S$, then also $\tau(\sigma(S),\sigma(S^\prime))$ holds.

\end{enumerate}

\end{postulate}

In the following we write $f_S$ to denote the interpretation of the function symbol $f \in \Sigma$ in the state $S$. Though we still have only necessary conditions for nd-seq algorithms, one can define further notions that are important for the development of the theory.

\begin{definition}\rm

A {\em location} of the nd-seq algorithm $\mathcal{A}$ is a pair $\ell = (f,(a_1,\dots,a_n))$ with a function symbol $f \in \Sigma$ of arity $n$ and all $a_i \in B$. If $B$ is the base set of state $S$ and $f_S(a_1,\dots,a_n) = a_0$ holds, then $a_0$ is called the {\em value} of the location $\ell$ in state $S$.

\end{definition}

We write $\text{val}_S(\ell)$ for the value of the location $\ell$ in state $S$. The {\em evaluation function} val can be extended to ground terms in a straightforward way.

\begin{definition}\rm

The {\em set of ground terms} over the signature $\Sigma$ is the smallest set $\mathbb{T}$ such that $f(t_1,\dots,t_n) \in \mathbb{T}$ holds for all $f \in \Sigma$ with $\text{ar}(f) = n$ and $t_1,\dots,t_n \in \mathbb{T}$\footnote{Clearly, for the special case $n = 0$ we get $f() \in \mathbb{T}$. Instead of $f()$ we usually write simply $f$.}. The {\em value} $\text{val}_S(t)$ of a term $t = f(t_1,\dots,t_n) \in \mathbb{T}$ in a state $S$ is defined by $\text{val}_S(t) = f_S(\text{val}_S(t_1),\dots,\text{val}_S(t_n))$.

\end{definition}

With the notions of location and value one can further define updates and their result on states\footnote{Note that update sets as we use them are merely differences of states.}.

\begin{definition}\label{def-update}\rm

An {\em update} of an nd-seq algorithm $\mathcal{A}$ in state $S$ is a pair $(\ell,v)$ with a location $\ell$ and a value $v \in B$, where $B$ is the base set of $S$. An update $(\ell,v)$ is {\em trivial} iff $\text{val}_S(\ell) = v$ holds. An {\em update set} is a set of updates. An update set $\Delta$ in state $S$ is {\em consistent} iff $(\ell,v_1), (\ell,v_2) \in \Delta$ implies $v_1 = v_2$, i.e. there can be at most one non-trivial update of a location $\ell$ in a consistent update set. If $\Delta$ is a 
consistent\footnote{Otherwise, usually the term 
	$S + \Delta$ used to define the successor state is considered as not defined. An alternative is to extend this definition letting 
$S + \Delta = S$, if $\Delta$ is inconsistent.} 
update set in state $S$, then $S + \Delta$ denotes the unique state $S^\prime$ with $\text{val}_{S^\prime}(\ell) = \begin{cases} v &\text{if}\; (\ell,v) \in \Delta \\ \text{val}_S(\ell) &\text{otherwise} \end{cases}$.

\end{definition}

Considering the locations, where a state $S$ and a successor state $S^\prime$ differ, gives us the following well-known fact (see \cite{gurevich:tocl2000}).

\begin{fact}\label{fact1}
	
If $\tau(S,S^\prime)$ holds, then there exists a unique minimal consistent update set $\Delta$ with $S + \Delta = S^\prime$.\footnote{The conclusion is true for any given pair $(S,S^\prime)$ of states, independently of the relation $\tau(S,S^\prime)$.}

\end{fact}

We use the notation $\Delta(S,S^\prime)$ for the consistent update set that is defined by $\tau(S,S^\prime)$. We further write $\boldsymbol{\Delta}(S)$ for the set of all such update sets defined in state $S$, i.e. $\boldsymbol{\Delta}(S) = \{ \Delta(S,S^\prime) \mid \tau(S,S^\prime) \}$.

The third postulate concerns bounded exploration. It is motivated by the simple observation that any algorithm requires a finite representation, which implies that only finitely many ground terms may appear in the representation, and these must then already determine the successor state---for a more detailed discussion see \cite{gurevich:tocl2000}---or the successor states in the case of non-determinism. Formally, this requires a notion of {\em coincidence} for a set of ground terms in different states.

\begin{definition}\rm

Let $T \subseteq \mathbb{T}$ be a set of ground terms for a nd-seq algorithm $\mathcal{A}$. Two states $S_1$ and $S_2$ with the same base set $B$ {\em coincide} on $T$ iff $\text{val}_{S_1}(t) = \text{val}_{S_2}(t)$ holds for all terms $t \in T$.

\end{definition}

\begin{postulate}[Bounded Exploration Postulate]\label{p-bound}\rm

Each nd-seq algorithm $\mathcal{A}$ comprises a finite set of ground terms $W \subseteq \mathbb{T}$ such that whenever two states $S_1$ and $S_2$ with the same base set coincide on $W$ the corresponding sets of update sets for $S_1$ and $S_2$ are equal, i.e. we have $\boldsymbol{\Delta}(S_1) = \boldsymbol{\Delta}(S_2)$. The set $W$ is called a {\em bounded exploration witness}.

\end{postulate}

Bounded exploration witnesses are not unique. In particular, the defining property remains valid, if $W$ is extended by finitely many terms. Therefore, without loss of generality we may tacitly assume that a bounded exploration witness $W$ is always closed under subterms. We then call the elements of $W$ {\em critical terms}. If $t$ is a critical term, then its value $\text{val}_S(t)$ in a state $S$ is called a {\em critical value}. This gives rise to the following well-known fact.

\begin{fact}\label{fact2}

The set $\boldsymbol{\Delta}(S)$ of update sets of an nd-seq algorithm $\mathcal{A}$ in a state $S$ is finite, and every update set $\Delta(S,S^\prime) \in \boldsymbol{\Delta}(S)$ is also finite.

\end{fact}

For a proof we first need to show that in every update $((f,(v_1,\dots,v_n)),v_0)$ in an update set $\Delta(S,S^\prime)$ the values $v_i$ are critical \cite{gurevich:tocl2000}. As $W$ is finite, there are only finitely many critical values, and we can only build finite update sets $\Delta(S,S^\prime)$ and only finitely many sets of update sets with these. We will use such arguments later in Section \ref{sec:recursion} to show that recursive algorithms are captured by recursive ASMs, and dispense with giving more details here.

\subsection{Recursion Postulate}\label{sect:recPostulate}

As remarked initially, an essential property of any recursive algorithm is the ability to perform call steps, i.e. to trigger an instance of a given
algorithm (maybe of itself) and remain waiting until the callee has computed an output for the given input. We make this explicit by extending the postulate on the one-step transition relation $\tau$ of nd-seq algorithms by characteristic conditions for a call step (see Postulate \ref{p-callStep} below).

Furthermore, it seems to be characteristic for runs of recursive algorithms that 
in a given state, the caller may issue in one step more than one call, 
though only finitely many, of callees which perform
their subcomputations independently of each other. For an example see the 
$sort$ rule in the \textit{mergesort\/} algorithm in Section \ref{sec:bsp}. The resulting `asynchronous parallelism' implies that the states in runs of a recursive algorithm are built over the union of the signatures of the calling and the called algorithms.

The independency condition for parallel computations of different instances of the given algorithms requires that for different calls, in particular for different calls of the same algorithm, the state spaces of the triggered subcomputations are separated from each other.
Below we make the term {\em instance of an algorithm} 
more precise to capture the needed encapsulation of subcomputations. This must be coupled with an appropriate input/output relation between the input provided by the caller and the output computed by the callee for this input, which will be captured by a {\em call relationship} in Definition \ref{def-call-rel}.  

This explains the following definition of an i/o-algorithm as nd-seq algorithm with call steps and distinguished function symbols for input and output.

\begin{definition}\label{def-io-alg}\rm
	
An {\em algorithm with input and output} (for short: {\em i/o-algorithm}) is an nd-seq algorithm whose one-step transition relation $\tau$ may comprise call steps satisfying the Call Step Postulate \ref{p-callStep} formulated below and whose signature $\Sigma$ is the disjoint union of three subsets
\[\Sigma = \Sigma_{in} \cup \Sigma_{loc} \cup \Sigma_{out}\]
containing respectively input, local and output function symbols that satisfy the input/output assumption defined below. 
	
\end{definition}

Function symbols in $\Sigma_{in}$, $\Sigma_{out}$ and $\Sigma_{loc}$, respectively, are called {\em input}, {\em output} and {\em local} function symbols. Correspondingly, locations with function symbol in  $\Sigma_{in}$, $\Sigma_{out}$ and $\Sigma_{loc}$, respectively, are called {\em input}, {\em output} and {\em local locations}. 
We include into input resp. output locations also variables which appear as input resp. output parameters of calls, although they are not function symbols.

The assumption on input/output locations of i/o-algorithms
is not strictly needed, but it can always 
be arranged and it eases the development of the theory.

\paragraph{\bf Input/Output Assumption} for i/o algorithms $\mathcal{A}$:

 \begin{enumerate}
	
\item Input locations of $\mathcal{A}$ are only read by $\mathcal{A}$, but never updated by $\mathcal{A}$. Formally, this implies that if $(\ell,v)$ is an update in an update set $\Delta(S,S^\prime)$ of $\mathcal{A}$ in any state $S$, then the function symbol $f$ in $\ell$ is not in $\Sigma_{in}$ of $\mathcal{A}$.
	
\item Output locations of $\mathcal{A}$ are never read by 
	$\mathcal{A}$, but can be written by $\mathcal{A}$. 
	This can be formalised by requiring that if $W$ is a bounded exploration witness, then for any term $f(t_1,\dots,t_n) \in W$ we have $f \notin \Sigma_{out}$.
	
\item Any initial state of $\mathcal{A}$ only depends on its input locations, so we may assume that $val_{S_0}(\ell) = \textit{undef\/}$ holds in every initial state $S_0$ of $\mathcal{A}$ for all output and local locations $\ell$.
	This assumption guarantees that when an i/o-algorithm is called, its run is initialized by the given input, which reflects the common intuition using input and output. 
	
\end{enumerate}

In a call relationship we call the caller the {\em 
	parent} and the callee the {\em child} algorithm. Intuitively, 
	
\begin{enumerate}\renewcommand{\labelenumi}{\alph{enumi})}

\item the parent algorithm is able to update input locations of 
the child algorithm, which determines the child's initial state;

\item when the child algorithm is called, control is handed 
over to it until it reaches a final state, in which state 
the parent takes back control and is able to read the output 
locations of the child;

\item the two algorithms have no other common locations. 

\end{enumerate}

Therefore we define:

\begin{definition}\label{def-call-rel}\rm	

A {\em call relationship} holds for (instances of) two i/o-algorithms $\mathcal{A}^p$ (parent) and $\mathcal{A}^c$ (child) if and only if they satisfy the following:	

\begin{itemize}
	
\item $\Sigma^{\mathcal{A}^c}_{in} \subseteq \Sigma^{\mathcal{A}^p}$. Furthermore,
	$\mathcal{A}^p$ may update input locations of $\mathcal{A}^c$, but never reads these locations. Formally this implies that for a bounded exploration witness $W$ of $\mathcal{A}^p$ and any term $f(t_1,\dots,t_n) \in W$ we have $f \notin \Sigma^{\mathcal{A}^c}_{in}$.
	
\item $\Sigma^{\mathcal{A}^c}_{out} \subseteq \Sigma^{\mathcal{A}^p}$.  Furthermore,
	$\mathcal{A}^p$ may read but never updates output locations of $\mathcal{A}^c$, so we have that for any update in an update set $\Delta(S,S^\prime)$ in any state $S$ of $\mathcal{A}^p$, its function symbol is not in $\Sigma^{\mathcal{A}^c}_{out}$.
	
\item $\Sigma^{\mathcal{A}^c}_{loc} \cap \Sigma^{\mathcal{A}^p} = \emptyset$ (no other common locations).

\end{itemize}	

\end{definition}

\begin{postulate}[Call Step Postulate]\label{p-callStep}\rm

When an i/o-algorithm $p$---the caller, viewed as parent algorithm---calls a finite number of i/o-algorithms $c_1,\ldots,c_n$---the callees, viewed as child algorithms $CalledBy(p)$---a call relationship  (denoted as $CalledBy(p)$) holds between the caller and each callee. The caller activates a fresh instance of each callee $c_i$ so that they can start their computations. These computations are independent of each other and the caller remains waiting---i.e. performs no step---until every callee has terminated its computation (read: has reached a final state). For each callee, the initial state of its computation is determined only by the input passed by the caller; the only other interaction of the callee with the caller is to return in its final state an output to $p$.

\end{postulate}

\begin{definition}\label{def-recAlg}\rm

A {\em sequential recursive algorithm} $\mathcal{R}$ is a finite set of i/o-algorithms---i.e. satisfying the branching time, abstract state, bounded exploration and call step postulates \ref{p-time}, \ref{p-state}, \ref{p-bound} and \ref{p-callStep}---one of which is distinguished as {\em main} algorithm.
The elements of $\mathcal{R}$ are also called components of  $\mathcal{R}$.

\end{definition}

Differently from runs of a nd-seq algorithm as defined by 
Definition \ref{def-seqRun}, where in each state at most one 
step of the nd-seq algorithm is performed, in a recursive 
run a sequential recursive algorithm $\mathcal{R}$ can perform in one 
step simultaneously one step of each of finitely many not 
terminated and not waiting called instances of its 
i/o-algorithms. This is expressed by the recursive run postulate \ref{p-recRun} below. In this postulate we refer to $Active$ and not $Waiting$ instances of components, which are defined as follows:
	
\begin{definition}\label{Waiting}\rm

To be $Active$ resp. $Waiting$ in a state $S$ is defined as follows:		
\begin{asm}
Active(q)  \IFF q \in Called \AND ~\NOT Terminated(q) \\
Waiting(p) \IFF
  ~ \FORSOME c \in CalledBy(p) ~ Active(c) \\
  Called=\{main\} \cup \bigcup_p CalledBy(p)       	
\end{asm}
$Called$ collects the instances of algorithms that are called during the run. $CalledBy(p)$ denotes the subset of $Called$ which contains all the children called by $p$. $Called=\{main\}$ and $CalledBy(p)=\emptyset$ are true in the initial state $S_0$, for each i/o-algorithm $p \in \mathcal{R}$. In particular, in $S_0$ the original component $main$ is considered to not be $CalledBy(p)$, for any $p$.
		
\end{definition}

\begin{postulate}[Recursive Run Postulate]\label{p-recRun}\rm
		
For a sequential recursive algorithm $\mathcal{R}$ with main component $main$ a {\em recursive run} is a sequence $S_0,S_1,S_2, \ldots$ of states together with a sequence $C_0,C_1,C_2, \ldots$ of sets of instances of components of $\mathcal{R}$ which satisfy the following constraints concerning the recursive run and bounded call tree branching:

\begin{description}

\item[\bf Recursive run constraint.] \

\begin{itemize}

\item $C_0$ is the singleton set $C_0=\{main\}$, i.e. every run starts with $main$,

\item  every $C_i$ is a finite set of in $S_i$ $Active$ and not $Waiting$ instances of components of $\mathcal{R}$, 

\item every $S_{i+1}$ is obtained in one  $\mathcal{R}$-step by performing in $S_i$ simultaneously one step of each i/o-algorithm in $C_i$.\footnote{Our reviewers worried here about synchronous, asynchronous and interleaving executions. Due to the independence of all the instances, whether they are executed in parallel, asynchronously or interleaved does not matter here. Note that $C_i$ is a not furthermore restricted finite \emph{subset} of all in $S_i$ $Active$ and not $Waiting$ (completely independent) instances of components of $\mathcal{R}$.  If $C_i$ contains more than one instance, the instances in this $C_i$ and only these are synchronized (independently for each $i$), but also interleaving (where each $C_i$ is a singleton set) and asynchronous execution are possible. We impose no constraint at all on how the sets $C_i$ are determined, e.g. by an external scheduling mechanism.} Such an $\mathcal{R}$-step is also called a {\em recursive step} of $\mathcal{R}$.

\end{itemize} 	

\item[\bf Bounded call tree branching.] There is a fixed natural number $m>0$, depending only on $\mathcal{R}$, which in every $\mathcal{R}$-run bounds the number of callees which can be called by a call step.

\end{description}

\end{postulate}

To capture the required independence of callee computations we now describe a way to make the concept of an instance of an algorithm and its computation more precise. The idea is to use for each callee a different state space, with the required connection between caller and callee through input and output terms. One can define an instance of an algorithm $\mathcal{A}$ by adding a label $a$, which we invite the reader to view as an agent executing the instance $\mathcal{A}_a=(a,\mathcal{A})$ of $\mathcal{A}$. The label $a$ can be used as environment parameter for the evaluation $val_S(t,a)$ of a term $t$ in state $S$ with the given environment. This yields different functions $f_a,f_{a^\prime}$ as interpretation of the same function symbol $f$ for different agents $a,a^\prime$, so that the run-time interpretations of a common signature element $f$ can be made to differ for different agents, due to different inputs which determine their initial states.\footnote{The idea underlies the definition of ambient ASMs we will use in the following. It allows one to classify certain $f$ as ambient-dependent functions, whereby the algorithm instances become context-aware. For the technical details we refer to  the definition in the textbook~\cite[Ch.4.1]{BoeRas18}.}

This allows us to make the meaning of `activating a fresh 
instance of a callee' in the Call Step Postulate more 
precise by using as fresh instance of a child algorithm 
$\mathcal{A}$ called by $p$ an instance $\mathcal{A}_c$ with a new 
label $c$, where the interpretation $f_c$ of each input 
or output function $f$ satisfies $f_c=f_p$ during the run of $\mathcal{A}_c$. Note that by the call relationship constraint in the Call Step Postulate,
input/output function symbols are in the signature of both the parent and the child algorithm.
This provides the ground for the `asynchronous parallelism' 
of independent subcomputations in the run constraint of the 
recursive run postulate. In fact, when a state $S^\prime$ is 
obtained from state $S$ by one step of each of finitely many 
$Active$ and not $Waiting$ i/o-algorithms $q_1,\ldots,q_k$, 
this means that for each $j \in \{ 1 ,\dots, k \}$ the one-step 
transition relation holds for the corresponding state 
restrictions, namely 
$\tau_{q_j}(res(S,\Sigma^{q_j}),res(S^\prime,\Sigma^{q_j}))$ 
where $res(S,\Sigma)$ denotes the restriction of state $S$ 
to the signature $\Sigma$.

With the above definitions one can make the Call 
Step Postulate more explicit by saying that if
$\mathcal{A}_p$ calls $\mathcal{A}^1,\ldots,\mathcal{A}^n$ in a state $S$ so that as a result 
$\tau_{\mathcal{A}_p}(S,S^\prime)$ holds\footnote{To simplify the 
	presentation we adopt a slight abuse of 
	notation, writing $\tau_{\mathcal{A}}(S,S^\prime)$ with the global states $S,S^\prime$ even where $\tau_{\mathcal{A}}$ really holds for their restriction to the sub-signature of the concrete algorithm $\mathcal{A}$.}, then for fresh instances 
$\mathcal{A}^i_{c_i}$ of $\mathcal{A}^i$ with input locations $input_i$ ($1 \leq i \leq n$) the following holds: 
\begin{asm}
	\mbox{In } S^\prime \mbox{ the following is true }\FORALL 1 \leq i \leq n:	 \+
	\mathcal{A}^i_{c_i} \in  Called \AND 
	\mathcal{A}^i_{c_i} \in CalledBy(\mathcal{A}_p) \AND  
	Initialized(\mathcal{A}^i_{c_i} ,input_i) \+
	\AND Waiting(\mathcal{A}_p)\footnote{Except the trivial case that all $\mathcal{A}^i_{c_i}$ when $Called$ in $S^{\prime}$ are already $Terminated$.}
\end{asm}
The predicate $Initialized(\mathcal{A}^i_{c_i},input_i)$ expresses that the restriction $res(S',\Sigma^{\mathcal{A}^i_{c_i}})$ of $S'$ to the signature of $\mathcal{A}^i_{c_i}$ is an initial state of $\mathcal{A}^i_{c_i}$ determined by $input_i$, so that $\mathcal{A}^i_{c_i}$ is ready to start its computation.

\begin{remark*}[\bf on Call Trees]
If in a recursive 
$\mathcal{R}$-run the main algorithm calls some
i/o-algorithms, this
call creates a finitely branched call tree whose nodes are 
labeled by the  instances of the i/o-algorithms involved, 
with active and not 
waiting algorithms labeling the leaves and with the main (the parent) 
algorithm labeling the root of the tree and becoming waiting. When the 
algorithm at a leaf makes a call, this extends the tree 
correspondingly. 
When the algorithm at a child of a node has terminated its 
computation, we delete the child from the tree. The leaves of this (dynamic) call tree are 
labeled by the active not waiting algorithms in the run. 
When the main algorithm terminates, the call tree is reduced again to the root labeled by the initially called {\em main algorithm}. 
\end{remark*}

Usually, it is expected that for recursive $\mathcal{R}$-runs each called i/o-algorithm reaches a final state, but in general it is not excluded that this is not the case. An example of the former case is given by \textit{mergesort\/}, whereas an example for the latter case is given by the recursive {\em sieve of Eratosthenes} algorithm discussed in \cite{moschovakis:2001} and used in Section \ref{sec:bsp} to illustrate our definitions.

\section{Capture of Recursive Algorithms}\label{sec:recursion}

We now proceed with the second step of our behavioural theory, the definition of an abstract machine model---these will be recursive ASMs, an extension of sequential ASMs---and the proof of the main theorem that the runs of this model capture the runs of sequential recursive algorithms (Theorems \ref{thm-plausible} and \ref{thm-capture}).

\subsection{Recursive Abstract State Machines}
\label{sect:recASM}

As common with ASMs let $\Sigma$ be a signature and let $\mathcal{U}$ be a universe of values. In addition, we assume a {\em background structure} comprising at least truth values and their connectives as well as the operations on them. Values defined by the background are assumed to be elements of $\mathcal{U}$. Then (ground) terms over $\Sigma$ are 
built in the usual way (using also the operations from the background), and they are interpreted taking $\mathcal{U}$ as base set---for details we refer to the standard definitions of ASMs \cite{boerger:2003}. This defines the set of states of {\em recursive ASM rules} we are going to 
define now syntactically. We proceed by induction, adding to the usual rules of non-deterministic sequential (nd-seq) ASMs (which we repeat here for the sake of completeness) named rules which can be called.\footnote{The terse definition here avoids complicated syntax. We tacitly permit parentheses to be used in rules when needed.} We use an arbitrary set $\mathcal{N}$ of names for named rules.

\begin{description}

\item[Assignment.] If $t_0, \dots, t_n$ are terms over the signature $\Sigma$ and $f \in \Sigma$ is a function symbol of arity $n$, then $f(t_1 ,\dots, t_n) := t_0$ is a  recursive ASM rule.

\item[Branching.] If $\varphi$ is a Boolean term over the signature $\Sigma$ and $r$ is a recursive ASM rule, then also \texttt{IF} $\varphi$ \texttt{THEN} $r$ is a recursive ASM rule.

\item[Bounded Parallelism.] If $r_1, \dots, r_n$ are  recursive ASM rules, then also their parallel composition, denoted \texttt{PAR} $r_1 \| \dots$ $\| r_n$ is a recursive ASM rule.

\item[Bounded Choice.] If $r_1, \dots, r_n$ are recursive ASM rules, then also the non-deterministic choice among them, denoted \texttt{CHOOSE} $r_1 \mid \dots \mid r_n$ is a recursive ASM rule.

\item[Let.] If $r$ is a recursive ASM rule and $t$ is a term and $x$ is a variable, then \texttt{LET} $x = t$ \texttt{IN} $r$ is also a recursive ASM rule.

\item[Call.] Let $t_0, \dots, t_n$ be terms where the outermost function symbol of $t_0$ is different from the outermost function symbol of $t_i$ for every $i \neq 0$. 
Let $N \in \mathcal{N}$ be the name of a rule of arity $n$, declared by $N(x_1,\ldots,x_n)=r$, where $r$ is a recursive ASM rule all of whose free variables are contained in $\{x_0,\ldots,x_n\}$. Then $t_0 \leftarrow N(t_1,\dots,t_n)$ is a recursive ASM  rule.

\end{description}

\begin{definition}\rm

A recursive ASM rule  of form $t_0 \leftarrow N(t_1,\dots,t_n)$ is called a {\em named i/o-rule} or simply i/o-rule.

\end{definition}

The same way a sequential recursive algorithm consists of finitely many i/o-algorithms, a recursive ASM $\mathcal{R}$ consists of finitely many recursive ASM rules, also called component (or component ASM) of $\mathcal{R}$.

\begin{definition}\label{def-rASM}\rm

A {\em recursive Abstract State Machine} (rec-ASM) $\mathcal{R}$ consists of a finite set of recursive ASM rules, one of which is declared to be the main rule. 

\end{definition}

For the signature $\Sigma$ of recursive ASM rules we use the notation $\Sigma_{in} \cup \Sigma_{loc} \cup \Sigma_{out}$ for the split of $\Sigma$ into the disjoint union of input, output and local functions. For  named i/o-rules $t_0 \leftarrow N(t_1,\dots,t_n)$ the outermost function symbol of $t_0$ is declared as an element of $\Sigma_{out}$ and for each $t_i$ the 
outermost function symbol of $t_i$ is declared as an element 
of $\Sigma_{in}$ ($i=1,\ldots,n$). In the definition of the semantics of a 
named i/o-rule we will take care that the input/output 
assumption and the call relationship defined in 
Section \ref{sect:recPostulate} for i/o-algorithms are satisfied by named i/o-rules. 

Sequential and recursive ASMs differ in their run concept, 
analogously to the difference between runs of an nd-seq 
algorithm and of a sequential recursive algorithm. A sequential ASM is a 
`mono-agent' machine: it consists of just one 
rule\footnote{For notational convenience, this rule is often 
	spelled out as a set of rules, however these rules are 
	always executed together, in parallel.} and in a 
sequential run this very same rule is applied in each 
step---by an execution agent that normally remains unmentioned. This 
changes with recursive ASMs which are `multi-agent' 
machines.  They consist of a set of independent rules, 
multiple instances of which (even of a same rule) may be 
called to be executed independently (for an example see the 
$sort$ rule in Sect.~\ref{sec:bsp}). We capture this by 
associating an execution agent $ag(r)$ with each rule $r$ so 
that each agent can execute its rule instance independently of the other agents, in its own copy of the state space 
(i.e. instances of states over the signature of the executed 
rule), taking into account the call relationship between caller and callee (see below). 

Therefore every single step of a recursive ASM  $\mathcal{R}$ 
may involve the execution of one step by each of finitely many $Active$ and not $Waiting$ agents $a$ which execute in their state space the rule $pgm(a)$ they are associated (we also say {\em equipped}) with.
To describe this separation of state spaces of different agents (in particular if they execute the same program),
we define {\em instances of a rule} $r$ by ambient ASMs of form $\AMB a \IN r$ with agents $a$ (see below for details). The following definition paraphrases the run constraint in the Recursive Run Postulate \ref{p-recRun}.

\begin{definition}\label{def-recAsmRun}\rm

A {\em recursive run} of a recursive ASM $\mathcal{R}$ is a sequence $S_0,S_1,S_2,\ldots$ of states together with a sequence $A_0,A_1,A_2,\ldots$ of subsets of $Agent$, where each $a \in Agent$ is equipped with a $pgm(a)$ that is an instance $\AMB a \IN r$ of a rule $r \in \mathcal{R}$, such that the following holds:

\begin{itemize}

\item $A_0$ is a singleton set $A_0=\{a_0\}$, which in $S_0$ equals the set $Agent$, and its agent $a_0$ is equipped with $pgm(a_0)= ~(\AMB a_0 \IN main)$.

\item $A_i$ is a finite set of in $S_i$ $Active$ and not $Waiting$ agents. We define (see Definition~\ref{Waiting}):
	\begin{asm}
		Active(a) \mbox{ (in state } S)    \IFF
		     a \in Agent \AND ~\NOT Terminated(pgm(a)) 
		     \mbox{ (in } S)   \\
		Waiting(a) \mbox{ (in state } S) \IFF ~ 
		\FORSOME a^\prime \in CalledBy(a) ~ Active(a^\prime)  \mbox{ (in } S)             	
	\end{asm} 
\item $S_{i+1}$ is obtained from $S_i$ in one $\mathcal{R}$-step by performing for each agent $a \in A_i$ one step of $pgm(a)$.\footnote{If one wants to stick to interleaving executions, it suffices to determine $A_i$ as singleton sets.}

\end{itemize} 

\end{definition}

To complete the definition of recursive ASM runs, which extends the notion of runs of sequential ASMs, 
it suffices (besides explaining ambient ASMs) to add a definition for what it means semantically to apply a named i/o-rule. Using the ASM framework this boils down to extend the inductive definition of the update sets computed by sequential ASMs in a given state by defining the update sets computed by named i/o-rules. 
 
A detailed definition of ambient ASMs can be found in~\cite[Ch.4.1]{BoeRas18}. Here it suffices to say that 
using $\AMB a \IN r$ as instance of a called rule $r$ permits to 
isolate the state space of agent $a$ from that of other 
agents, namely by evaluating terms $t$ in state $S$ 
considering also the agent parameter $a$, using $val_S(t,a)$ 
instead of $val_S(t)$.
To establish the call relationship we 
require below the following:  when a recursive ASM rule $r$, executed by a 
parent agent $p$, calls a rule $q$ to be executed by 
a child agent $c$, then the input/output functions $f$ of $q$ are also functions in $r$ and are interpreted there in the state space of $p$ the same way as in the state space of $c$. 

For the sake of completeness we repeat 
the definition of update sets for sequential ASM rules from \cite{gurevich:lipari1995} and extend it for named i/o-rules.
Rules $r$ of sequential ASMs do not change neither the set $Agent$ nor the $pgm$ function, so $Agent$ and $pgm$ do not appear in the definition of $\Delta_r(S)$.\footnote{Note that $\Delta_r(S)$ defines the set of update sets by which rule $r$ changes state $S$ into a successor state $S^\prime$.}  $i/o$-rules are the only rules which involve also introducing a new element $a$ into $Agent$ (with a value assigned to $pgm(a)$) and a state initialization corresponding to the provided input, so that $a$ executes its instance of the called rule.

\begin{itemize}

\item If $r$ is an assignment rule $f(t_1 ,\dots, t_n) := t_0$, then let $v_i = \text{val}_S(t_i)$. We define $\Delta_r(S) = \{ \{ ((f,(v_1,\dots,v_n)),v_0) \} \}$.

\item If $r$ is a branching rule \texttt{IF} $\varphi$ \texttt{THEN} $r^\prime$, then let $v$ be the truth value $\text{val}_S(\varphi)$. We define $\Delta_r(S) = \Delta_{r^\prime}(S)$ for $v = \textbf{true}$ and $\Delta_r(S) = \{ \emptyset \}$ otherwise.

\item If $r$ is a parallel composition rule \texttt{PAR} $r_1 \| \dots \| r_n$,\footnote{Parallel composition rules are also written by displaying the components $r_i$ vertically, omitting \texttt{PAR} and $\|$.} then we define $\Delta_r(S) = \{ \Delta_1 \cup\dots\cup \Delta_n \mid \Delta_i \in \Delta_{r_i}(S) \}$.

\item If $r$ is a bounded choice rule \texttt{CHOOSE} $r_1 \mid \dots \mid r_n$, then we define $\Delta_r(S) = \Delta_{r_1}(S) \cup\dots\cup \Delta_{r_n}(S)$.

\item If $r$ is a let rule \texttt{LET} $x = t$ \texttt{IN} $r^\prime$, then let $v = \text{val}_S(t)$, and define $\Delta_r(S) = \{ [v / x] . \Delta \mid [v / x] . \Delta \in \Delta_{r^\prime}(S)\}$.

\end{itemize}

Now consider the case that $r$ is a call rule $t_0 \leftarrow N(t_1,\dots,t_n)$. In this case let $ t_0 = f(t_1^\prime,\dots,t_k^\prime)$, and let $N(x_1,\dots,x_n)=q$ be the declaration of the rule named $N$, with all free variables of $q$ among $x_0,x_1,\ldots,x_n$.

In the call tree, the caller program $r$ plays the role of the parent of the called child program that will be executed by a new agent $c$. The child program is an instance $q_c$ of $q$ with the outer function symbols of $t_i$ for $1 \leq i \leq n$ classified as denoting input functions (which are not read by the caller program) and with the outer function symbol $f$ of $t_0$ classified as denoting an output function (which is not updated by the caller program).\footnote{The input parameters and the output location parameters are passed by value, so that the involved i/o-function symbols can be considered as belonging to the signature of caller and callee.} The first two of the call relationship conditions are purely syntactical and can be assumed (without loss of generality) for caller and callee programs. 
The third condition is satisfied, because each local function symbol $f$ of arity $n$ is implicitly turned in a program instance  into an ($n+1$)-ary function symbol, namely by adding the additional agent as environment parameter for the evaluation of terms with $f$ as leading function symbol. Therefore, each local function of the callee is different from each local function of the caller, and to execute the 
call rule means to create a new agent $c$,\footnote{The function $\NEW$ is assumed to yield for each invocation a fresh element, pairwise different ones for parallel invocations. One can define such a function also by an $\IMPORT$ construct which operates on a (possibly infinite) special {\em reserve} set and comes with an additional constraint on the $\PAR$ construct to guarantee that parallel imports yield pairwise different fresh elements, see \cite[2.4.4]{boerger:2003}.} which is 
$CalledBy$ the agent $\SELF$ that executes the call, to equip $c$ with the fresh program
instance $q_c$ and $\ASM{Initialize}$ its state by the values 
of $t_i,t_j^\prime$. This makes the callee ready to run 
and puts the caller into $Waiting$ mode, in the sense 
defined by Definition \ref{Waiting} (except the trivial case that $q_c$ 
is already $Terminated$ when $Called$ so that it will not be executed).

In other words we define $\Delta_r(S)$ as the singleton set containing the update set computed in state $S$ by the following ASM, a rule we denote by
$\ASM{Call}(t_0 \leftarrow N(t_1,\dots,t_n))$ which interpretes the named i/o-rule $t_0 \leftarrow 
N(t_1,\dots,t_n)$.

\begin{definition}\label{def-CallNamedRule}\rm

\begin{asm}
\ASM{Call}(t_0 \leftarrow N(t_1,\dots,t_n))=\+
\LET N(x_1,\dots,x_n)=q \mbox{  // declaration of }N\\
\LET  v_1 = t_1 ,\ldots,  v_n =t_n  
    \mbox{  // input evaluation  
	            $\text{val}_S(t_i,\SELF)$ by caller}\\
\LET t_0=f(t_1^\prime,\dots,t_k^\prime)\\		      
\LET  v_1^\prime = t_1^\prime,\ldots,  
           v_k^\prime = t_k^\prime
     \mbox{  // output location evaluation  
     	$\text{val}_S(t_i^\prime,\SELF)$ by caller}      \\
\LET c=~\NEW(Agent)	\+
    pgm(c):= ~ \AMB c \IN q  \mbox{  // equip callee with its program instance}\\
    \ASM{Insert}(c,CalledBy(\SELF)) \\
    \ASM{Initialize}(q_c,v_1 / x_1,\ldots,v_n /     
         x_n,f(v_1^\prime,\ldots, v_k^\prime)/x_o ) \\
    CalledBy(c):=\emptyset \-  
\end{asm}
\end{definition}

Note that $(f,(v_1^\prime,\ldots, v_k^\prime))$ denotes the output location which the caller expects to be updated by the callee with the return value.

\begin{theorem}\label{thm-plausible}

Each recursive ASM $\mathcal{M}$ defines a sequential recursive algorithm $\mathcal{M}_a$ (in the sense of Definition \ref{def-recAlg}) such that the recursive runs of $\mathcal{M}$ can be step-for-step simulated by the runs of $\mathcal{M}_a$.

\end{theorem}

\begin{proof} Let $\mathcal{M}$  be a recursive ASM. First of all we have to show that it satisfies the postulates in Definition \ref{def-recAlg} whereby it is a sequential recursive algorithm. 

Each rule $r$ belonging to $\mathcal{M}$, including named i/o-rules, is associated with a signature $\Sigma^r$ given by the function symbols that appear in the rule or in the rule body if the rule is a named rule. This 
together with the agents $c$ in $\AMB c \IN r$, defines the states (as sets of states of signature $\Sigma^r$, one per $r \in \mathcal{M}$) and gives the satisfaction of the abstract state postulate \ref{p-state}.

The satisfaction of the branching time postulate \ref{p-time} is an immediate consequence of the fact that for every state $S$, applying any recursive ASM rule in $S$, including named i/o-rules, yields a successor state. 

For the satisfaction of the bounded exploration postulate \ref{p-bound} we exploit that by \cite{gurevich:tocl2000}, each sequential ASM (i.e. without named i/o-rules) which appears in $\mathcal{M}$ is an nd-seq algorithm and thus satisfies the bounded exploration postulate. To extend this to the rules of $\mathcal{M}$, for each (of the finitely many) named i/o-rule we take every bounded exploration witness which appears in the rule body (including the parameters). By Definition \ref{def-CallNamedRule} these witnesses determine 
the update sets yielded by the named i/o-rule in any given state.

By the definition of recursive ASM runs (Definition~\ref{def-recAsmRun}) and of the effect of a call rule step (Definition \ref{def-CallNamedRule}), the call step postulate \ref{p-callStep} is satisfied by every recursive ASM.

As to the recursive run postulate \ref{p-recRun}, the run constraint is satisfied by the definition of recursive ASM runs (Definition \ref{def-recAsmRun}). The bounded call tree branching constraint is satisfied, because there are only finitely many named i/o-rules in each of the finitely many rules $r \in \mathcal{M}$.

It remains to show that the recursive runs of $\mathcal{M}$ can be step-for-step simulated by corresponding runs of the sequential recursive algorithm $\mathcal{M}_a$ that is induced by this interpretation of $\mathcal{M}$. This follows from the two run characterizations in Postulate \ref{p-recRun} and Definition \ref{def-recAsmRun} and from the fact that the successor relation of $\mathcal{M}_a$ is defined by the update sets which are yielded by the rules of $\mathcal{M}$ and define also the successor relation of $\mathcal{M}$.\hfill
\end{proof}

\subsection{The Characterisation Theorem}

We now show the converse of Theorem \ref{thm-plausible}. The proof largely follows the ideas underlying the proof of the sequential ASM thesis in \cite{gurevich:tocl2000}.

\begin{theorem}\label{thm-capture}

For each sequential recursive algorithm $\mathcal{R}$ in the sense of Definition \ref{def-recAlg} there exists a recursive Abstract State Machine which is equivalent to $\mathcal{R}$ with respect to recursive runs (in the sense of Definition \ref{def-recAsmRun} and Postulate \ref{p-recRun}).

\end{theorem}

\begin{proof}
Let $\mathcal{R}$ denote any sequential recursive algorithm. Then 
for each state $S$ and a successor state $S^\prime$ in a 
recursive run of $\mathcal{R}$ we obtain by Fact \ref{fact1} an update set 
$\Delta(S,S^\prime) \in \boldsymbol{\Delta}(S)$. According to the 
Recursive Run Postulate \ref{p-recRun} each such state 
transition is defined by one step of each of finitely many $Active$ and not $Waiting$ i/o-algorithms $\mathcal{A}_i$.
Each of these i/o-algorithms is a fresh instance $\mathcal{A}_i$ of some component of $\mathcal{R}$. In particular, 
by the freshness and the independence condition in the Call 
Step Postulate \ref{p-callStep}, the instances 
$\mathcal{A}_i$ have disjoint signatures $\Sigma_i$ and yield subruns 
with states $\text{res}(S,\Sigma_i)$ and update sets 
$\Delta(\text{res}(S,\Sigma_i),\text{res}(S^\prime,\Sigma_i))$. 

Consider now any such fresh instance $\mathcal{A}_i$ of a component $\mathcal{A} \in \mathcal{R}$. 
All function symbols used by $\mathcal{A}_i$ in its 
states and update sets are copies of function symbols of 
$\mathcal{A}$, labelled to ensure the freshness 
condition of the instance. Removing these labels we obtain for any state $S$ of $\mathcal{A}$ and successor state $S^\prime$ of $\mathcal{A}$ pairs $(S,\Delta_{\mathcal{A}}(S,S^\prime))$ with $\Delta_{\mathcal{A}}(S,S^\prime) \in \boldsymbol{\Delta}_{\mathcal{A}}(S)$. Let $D_{\mathcal{A}}$ be the set of all pairs $(S,\boldsymbol{\Delta}_{\mathcal{A}}(S))$ obtained this way, for any state $S$ in the given recursive run of  $\mathcal{R}$. 

To complete the proof of the theorem it therefore suffices to show that the result of any $\mathcal{A}$-step in the given run, namely to apply an update set  $\Delta_{\mathcal{A}}(S,S^\prime) \in \boldsymbol{\Delta}_{\mathcal{A}}(S)$ to a state $S$, can be described as result of a step of a recursive ASM rule $r_{\mathcal{A}}$, namely to apply to $S$ an update set in $\boldsymbol{\Delta}_{r_{\mathcal{A}}}(S)$. This is established by the following Lemma \ref{lem-rule}.\hfill
\end{proof}

\begin{lemma}\label{lem-rule}

For each $\mathcal{A}$ there exists a recursive ASM rule $r_{\mathcal{A}}$ with $\boldsymbol{\Delta}_{r_{\mathcal{A}}}(S) = \boldsymbol{\Delta}_{\mathcal{A}}(S)$ for all states $S$ appearing in $D_{\mathcal{A}}$.

\end{lemma}

\begin{proof}
We choose a fixed bounded exploration witness $W_{\mathcal{A}}$ for each $\mathcal{A} \in \mathcal{R}$.

First we show that the argument values of any location $\ell$  in an update of 
$\mathcal{A}$ in any state $S$ are critical values in $S$. The
proof uses the same argument as in \cite[Lemma~6.2]{gurevich:tocl2000}.

To show the property consider an arbitrary update set 
$\Delta_{\mathcal{A}}(S,S^\prime) \in \boldsymbol{\Delta}_{\mathcal{A}}(S)$ and let 
$(\ell,v_0) \in \Delta_{\mathcal{A}}(S,S^\prime)$ be an update at location $\ell = (f,(v_1,\dots,v_n)$. We show that the assumption that $v_i$ is not a critical value leads to a contradiction.

If $v_i$ is not a critical value, one can create a new structure $\hat{S}$ by swapping $v_i$ with a fresh value $w$ not appearing in $S$ (e.g. $w$ taken from the $Reserve$ set), so $\hat{S}$ is a state of $\mathcal{A}$. As $v_i$ is assumed to not being critical, we must have $val_S(t) = \text{val}_{\hat{S}}(t)$ for all terms $t \in W_\mathcal{A}$. Therefore it follows from the bounded exploration postulate that $\boldsymbol{\Delta}_{\mathcal{A}}(S) = \boldsymbol{\Delta}_{\mathcal{A}}(\hat{S})$. This implies that the update $((f,(v_1,\dots,v_n)),v_0)$ appears in at least one update set in $\Delta_{\mathcal{A}}(\hat{S})$ (since $(\ell,v_0) \in \Delta_{\mathcal{A}}(S,S^\prime) \in \boldsymbol{\Delta}_{\mathcal{A}}(S)$), contradicting the fact that $v_i$ does not occur in $\hat{S}$ and thus cannot occur in an update set created in this state. 
\medskip

Furthermore, for each pair $(S,\boldsymbol{\Delta}_{\mathcal{A}}(S)) \in D_{\mathcal{A}}$ we have a {\em recursion depth} function which indicates the maximal nesting of recursive calls performed starting in state $S$ to compute the value of an output location in a possible successor state $S^\prime$. The function is defined inductively as follows (induction on the call tree), taking the maximum over all successor states $S^\prime$ of $S$ and over all updates leading from $S$ to $S^\prime$.

\begin{itemize}

\item $\text{rdepth}_{\mathcal{A}}(S) = \max \{ \text{rdepth}_{\mathcal{A}}(S,S^\prime)) \mid \Delta_{\mathcal{A}}(S,S^\prime) \in \boldsymbol{\Delta}_{\mathcal{A}}(S) \}$
	
\item $\text{rdepth}_{\mathcal{A}}(S,S^\prime)) = \max \{ \text{depth}(\ell,v) \mid (\ell,v) \in \Delta_{\mathcal{A}}(S,S^\prime) \}$
	
\item  $\text{depth}(\ell,v)$ (with $\ell = (f,(v_1,\dots,v_n))$) defined as follows:
	
\begin{itemize}
		
\item Case 1: $f$ is an output function symbol of a terminating recursive subcomputation started in $S$ and leading to $S^\prime$. 

Formally this means that for some callee $\mathcal{A}^c$ just activated in the run by $\mathcal{A}$, the restriction  
		$\text{res}(S,\Sigma^c)$ of $S$ to the signature
		of  $\mathcal{A}^c$ is an initial state $S_0^c$
		of a terminating run  $S_0^c, S_1^c, \dots, S_k^c$ of the callee, during which $\mathcal{A}$ remains waiting, 
		and such that 
		\begin{itemize}
			\item 
		the callee receives in state $S$ the input from the caller, expressed by the equation $\text{res}(S_0^c,\Sigma^c_{in}) = \text{res}(S,\Sigma^c_{in})$, 
		\item the caller receives in state $S^\prime$ the callee's output in the callee's final state, formally $\text{res}(S^\prime,\Sigma^c_{out})$ $= \text{res}(S_k^c,\Sigma^c_{out})$ 
		for the final state $S_k^c$ for $\mathcal{A}^c$ 
	    \end{itemize}
        and $f \in \Sigma^c_{out}$. 
		
		Then the depth of the update is defined as successor of the maximal recursion depth between any two successive states in the subcomputation, formalized by the equation  $\text{depth}(\ell,v) = \max \{ \text{rdepth}_{\mathcal{A}^c}(S_i^c,S_{i+1}^c) \mid 0 \le i \le k-1 \} + 1$.
		
\item Case 2: Otherwise. Then we define $\text{depth}(\ell,v) = 0$.
		
\end{itemize}
	
\end{itemize}

\noindent
We now proceed by a case distinction for $\text{rdepth}_{\mathcal{A}}(S)$.

\paragraph{\bf Case 1:}

$\text{rdepth}_{\mathcal{A}}(S)$ is defined for all states $S \in \mathcal{S}_{\mathcal{A}}$. In this case we proceed by induction over $d = \max \{ \text{rdepth}_{\mathcal{A}}(S) \mid S \in \mathcal{S}_{\mathcal{A}} \}$. The base case is de facto the proof of the non-deterministic sequential ASM thesis.

\paragraph{\bf Induction Base:}

Let $d=0$. First we construct for every state $S$ with $\text{rdepth}_{\mathcal{A}}(S) = 0$ a (sequential ASM) rule $r_{\mathcal{A},S}$ whose application to $S$ yields the updates sets defined for $S$ by $\mathcal{A}$, formally such that $\boldsymbol{\Delta}_{r_{\mathcal{A},S}}(S) = \boldsymbol{\Delta}_{\mathcal{A}}(S)$, and the same for all states $\hat{S}$ that are `similar' to $S$ (as defined below) wrt the bounded exploration witness $W_{\mathcal{A}}$. 

To show this let $S$ be a state with $\text{rdepth}_{\mathcal{A}}(S) = 0$ and let $S^\prime$ be any successor state, resulting from $S$ by applying the updates in
$\Delta_{\mathcal{A}}(S,S^\prime) \in \boldsymbol{\Delta}_{\mathcal{A}}(S)$. Consider any
such update $(\ell,v_0) \in \Delta_{\mathcal{A}}(S,S^\prime)$ at location $\ell = (f,(v_1,\dots,v_n)$. As all $v_i$ are critical values and there is no child with $f \in \Sigma^c_{out}$, there exist terms $t_i \in W_{\mathcal{A}}$ with $\text{val}_S(t_i) = v_i$. Thus, the assignment rule $f(t_1,\dots,t_n) := t_0$ yields the given update in state $S$ and
the parallel composition $r_{\mathcal{A},S,S^\prime}$ of all these assignment rules for every update in $\Delta_{\mathcal{A}}(S,S^\prime)$  yields in $S$ the update set $\Delta_{\mathcal{A}}(S,S^\prime)$. Therefore the bounded choice composition of these rules for all successor states $S^\prime$ defines a rule $r_{\mathcal{A},S}$ with $\boldsymbol{\Delta}_{r_{\mathcal{A},S}}(S) = \boldsymbol{\Delta}_{\mathcal{A}}(S)$.

Next we extend $\boldsymbol{\Delta}_{r_{\mathcal{A},S}}(S) = \boldsymbol{\Delta}_{\mathcal{A}}(S)$ from $S$ to all states which are `similar' to $S$ with respect to the bounded exploration witness $W_{\mathcal{A}}$.\footnote{These cases are captured in Lemmata 6.7, 6.8. and 6.9 in \cite{gurevich:tocl2000}} 

\begin{enumerate}

\item First, we show the equation  $\boldsymbol{\Delta}_{r_{\mathcal{A},S}}(\hat{S}) = \boldsymbol{\Delta}_{\mathcal{A}}(\hat{S})$ for every state $\hat{S}$ which coincides with $S$ on $W_{\mathcal{A}}$. In fact, if $S$ and $\hat{S}$ coincide on $W_{\mathcal{A}}$, then $\boldsymbol{\Delta}_{\mathcal{A}}(S) = \boldsymbol{\Delta}_{r_{\mathcal{A},S}}(S) =\boldsymbol{\Delta}_{r_{\mathcal{A},S}}(\hat{S})$ holds because the rule $r_{\mathcal{A},S}$ only uses terms in $W_{\mathcal{A}}$, which have the same values in $S$ and $\hat{S}$. We further have $\boldsymbol{\Delta}_{\mathcal{A}}(S) = \boldsymbol{\Delta}_{\mathcal{A}}(\hat{S})$ due to the bounded exploration postulate. These equations together give $\boldsymbol{\Delta}_{r_{\mathcal{A},S}}(\hat{S}) = \boldsymbol{\Delta}_{\mathcal{A}}(\hat{S})$.

\item Second, we show that $\boldsymbol{\Delta}_{r_{\mathcal{A},S}}(S_1) = \boldsymbol{\Delta}_{\mathcal{A}}(S_1)$ carries over to isomorphic states $S_2$.  Let $S_1, S_2$ be isomorphic states for which $\boldsymbol{\Delta}_{r_{\mathcal{A},S}}(S_1) = \boldsymbol{\Delta}_{\mathcal{A}}(S_1)$ is true. Let $\zeta$ be the isomorphism with $\zeta S_2 = S_1$. Then we have  $\zeta \boldsymbol{\Delta}_{r_{\mathcal{A},S}}(S_2) = \boldsymbol{\Delta}_{r_{\mathcal{A},S}}(S_1)$ (because all ASMs satisfy the Abstract State Postulate) and also $\zeta \boldsymbol{\Delta}_{\mathcal{A}}(S_2) = \boldsymbol{\Delta}_{\mathcal{A}}(S_1)$  (by the Abstract State Postulate). These equations together give $\zeta \boldsymbol{\Delta}_{r_{\mathcal{A},S}}(S_2) = \zeta \boldsymbol{\Delta}_{\mathcal{A}}(S_2)$ and hence also $\boldsymbol{\Delta}_{r_{\mathcal{A},S}}(S_2) = \boldsymbol{\Delta}_{\mathcal{A}}(S_2)$.

\item Third, we conclude from (i) and (ii) that the equation $\boldsymbol{\Delta}_{r_{\mathcal{A},S}}(\hat{S}) = \boldsymbol{\Delta}_{\mathcal{A}}(\hat{S})$ holds for  
every state $\hat{S}$ that is $W_{\mathcal{A}}$-similar to $S$. We define states $S_1, S_2$ to be {\em $W_{\mathcal{A}}$-similar} iff they identify the same critical  terms, i.e. formally iff $val_{S_1}(t_1) = val_{S_1}(t_2)  \Leftrightarrow val_{S_2}(t_1) = val_{S_2}(t_2)$ holds for all terms $t_1,t_2 \in W_{\mathcal{A}}$.
No let $\hat{S}$ be any state that is $W_{\mathcal{A}}$-similar to $S$. We can assume without loss of generality that $S$ and $\hat{S}$ are disjoint. 

In fact, if this is not the case consider a state $\ddot{S}$ isomorphic to $\hat{S}$, in which each value that appears also in $S$ is replaced by a fresh one. Then $\ddot{S}$ is disjoint from $S$ and by construction $W_{\mathcal{A}}$-similar to $\hat{S}$, hence also $W_{\mathcal{A}}$-similar to $S$. 

Now define a structure $S^*$ isomorphic to $\hat{S}$ by replacing $\text{val}_{\hat{S}}(t)$ by  $val_S(t)$ for all $t \in W_{\mathcal{A}}$. The definition of $S^*$ is consistent because since $S$ and $\hat{S}$ are $W_{\mathcal{A}}$-similar, $\text{val}_S(t_1) = \text{val}_S(t_2) \Leftrightarrow \text{val}_{\hat{S}}(t_1) = \text{val}_{\hat{S}}(t_2)$ holds for all terms $t_1, t_2 \in W_{\mathcal{A}}$. Since $S$ and $S^*$ coincide on $W_{\mathcal{A}}$, we obtain by (i) that
$\boldsymbol{\Delta}_{r_{\mathcal{A},S}}(S^*) = \boldsymbol{\Delta}_{\mathcal{A}}(S^*)$ which by (ii) implies $\boldsymbol{\Delta}_{r_{\mathcal{A},S}}(\hat{S}) = \boldsymbol{\Delta}_{\mathcal{A}}(\hat{S})$.

\end{enumerate}

To complete the proof for the induction base we exploit that $W_{\mathcal{A}}$ is finite, hence there are only finitely many partitions of $W_{\mathcal{A}}$ and only finitely many $W_{\mathcal{A}}$-similarity classes $[S_i]_{W_{\mathcal{A}}}$ (say $i=1,\dots,m$). For each such class we define a formula $\varphi_i$ such that for each state $S$ this formula evaluates in $S$ to true if and only if $S \in [S_i]_{W_{\mathcal{A}}}$ holds.  $\varphi_i$ formalizes the similarity type of $S_i$ by the conjunction of all equations $t_1=t_2$ with $val_{S_i}(t_1)=val_{S_i}(t_2)$ and all inequalities $ t_1 \not= t_2$ with $val_{S_i}(t_1)\not =val_{S_i}(t_2)$
for all critical terms $t_1, t_2\in W_{\mathcal{A}}$. Then we can define the rule $r_{\mathcal{A}}$ as follows:

\begin{quote}
\texttt{PAR} (\texttt{IF} $\varphi_1$ \texttt{THEN} $r_{\mathcal{A},S_1}$) $\|$ \dots $\|$ (\texttt{IF} $\varphi_m$ \texttt{THEN} $r_{\mathcal{A},S_m}$)
\end{quote}

\paragraph{\bf Induction Step:}

Let $d>0$. For a state $S$ with $\text{rdepth}_{\mathcal{A}}(S) = 0$ we proceed as in the base case to construct a rule $r_{\mathcal{A},S}$ such that $\boldsymbol{\Delta}_{r_{\mathcal{A},S}}(\hat{S}) = \boldsymbol{\Delta}_{\mathcal{A}}(\hat{S})$ holds for all states $\hat{S}$ that are $W_{\mathcal{A}}$-similar to $S$. For the case that $\text{rdepth}_{\mathcal{A}}(S) > 0$ we construct such a rule as follows.

Let $(\ell,v_0) \in \Delta_{\mathcal{A}}(S,S^\prime) \in \boldsymbol{\Delta}_{\mathcal{A}}(S)$ be any update at location $\ell = (f,(v_1,\dots,v_n))$. Then we have two cases:

\begin{enumerate}

\item If there is no child algorithm $\mathcal{A}^c$ with $f \in \Sigma^c_{out}$, then we argue as in the base case, i.e. as all $v_i$ are critical values, there exist terms $t_i \in W_{\mathcal{A}}$ with $\text{val}_S(t_i) = v_i$, and thus the assignment rule $f(t_1,\dots,t_n) := t_0$ produces the given update in state $S$. 

\item If we have $f \in \Sigma^c_{out}$ for some child algorithm $\mathcal{A}^c$, i.e. $f$ is an output function symbol of an algorithm that is called by $\mathcal{A}$, then according to our assumption on call relationships, locations with such function symbols never appear in an update set created by $\mathcal{A}$ itself, so the update results from a final state of a run of $\mathcal{A}^c$. Let this run be $S_0^c, S_1^c, \dots$ with final state $S^c_k$.

Then we must have $\max \{ \text{rdepth}_{\mathcal{A}^c}(S) \mid S \in \mathcal{S}_{\mathcal{A}^c} \} \le d-1$, so we can apply the induction hypothesis. That is, there exists an ASM rule $r_{\mathcal{A}^c}$ with $\boldsymbol{\Delta}_{r_{\mathcal{A}^c}}(S_i^c) = \boldsymbol{\Delta}_{\mathcal{A}^c}(S^c_i)$ for all $0 \le i \le k$. 

\begin{itemize}

\item As the values $v_i$ for $i=1,\dots,n$ are critical in $S$, we find terms $\hat{t}_i \in W_{\mathcal{A}}$ with $\text{val}_S(\hat{t}_i) = v_i$.

\item As $v_0$ results from an update made by $\mathcal{A}^c$, it is critical in $S^c_\ell$ for some $\ell < k$, so there must exist a term $t_0 \in W_{\mathcal{A}^c}$ with $\text{val}_{S^c_\ell}(t_0) = v_0$.

\item The initial state $S^c_0$ is defined by values $v_i^\prime$ ($i=1,\dots,m$) of input locations, which are critical for $S^c_0$, which gives rise to terms $t_i^\prime \in W_{\mathcal{A}^c}$ with $\text{val}_{S^c_0}(t_i^\prime) = v_i^\prime$.

\item As the input values $v_i^\prime$ for $\mathcal{A}^c$ have been produced by updates made by $\mathcal{A}$ we further find terms $t_i \in W_{\mathcal{A}}$ with $\text{val}_S(t_i) = v_i^\prime$.

\end{itemize}

Using a name $N$ for the rule of $\mathcal{A}^c$ we obtain a named rule $f(\hat{t}_1,\dots,\hat{t}_n) \leftarrow N(t_1^\prime,\dots,t_m^\prime)$. Furthermore, according to the definition of update sets produced by call rules, we see that the call $t_0 \leftarrow N(t_1,\dots,t_m)$ produces the update $((f,(v_1,\dots,v_n),v_0)$.

\end{enumerate}

Again the parallel composition of all the assignment and call rules (for every update in $\Delta_\mathcal{A}(S,S^\prime)$) yields in $S$ the update set $\Delta_\mathcal{A}(S,S^\prime)$.The bounded choice composition of these rules for all successor states $S^\prime$ defines a rule $r_{\mathcal{A},S}$ with $\boldsymbol{\Delta}_{r_{\mathcal{A},S}}(S) = \boldsymbol{\Delta}_{\mathcal{A}}(S)$.

Using again the same arguments as in (i), (ii), and (iii) for the base case we get $\boldsymbol{\Delta}_{r_{\mathcal{A},S}}(\hat{S}) = \boldsymbol{\Delta}_{\mathcal{A}}(\hat{S})$ for all states $\hat{S}$ that are $W_{\mathcal{A}}$-similar to $S$, and exploiting the finiteness of bounded exploration witnesses we obtain again the rule $r_{\mathcal{A}}$ with the required property, i.e. $\boldsymbol{\Delta}_{r_{\mathcal{A}}}(S) = \boldsymbol{\Delta}_{\mathcal{A}}(S)$ for all states $S$.

\paragraph{\bf Case 2:} 

Assume that $\text{rdepth}_{\mathcal{A}}(S)$ is not defined for all states $S \in \mathcal{S}_{\mathcal{A}}$. 

For states $S$ for which $\text{rdepth}_{\mathcal{A}}(S)$ is defined we use the same construction as above to obtain a rule $r_{\mathcal{A},S}$ with $\boldsymbol{\Delta}_{r_{\mathcal{A},S}}(S) = \boldsymbol{\Delta}_{\mathcal{A}}(S)$.

Now take a state $S$, for which $\text{rdepth}_{\mathcal{A}}(S)$ is not defined. Then there exists a child algorithm $\mathcal{A}^c$ with an initial state $S^c_0$ initiated by $\mathcal{A}$, but no run starting in $S^c_0$ leads to a final state. Then the input values define terms $t_i$ and $t_i^\prime$ for $i=1,\dots,m$ in the same way as for the construction of the call rule above. We can use arbitrary critical terms $t_0$, $\hat{t}_i$ ($i=1,\dots,n$) and a named rule $f(\hat{t}_1,\dots,\hat{t}_n) \leftarrow N(t_1^\prime,\dots,t_m^\prime)$ as before, then the call rule $t_0 \leftarrow N(t_1,\dots,t_m)$ will lead to the run of $\mathcal{A}^c$ without final state.

The rule $r_{\mathcal{A}}$ with the required property, i.e. $\boldsymbol{\Delta}_{r_{\mathcal{A}}}(S) = \boldsymbol{\Delta}_{\mathcal{A}}(S)$ for all states $S$, then results by applying again the same arguments as above.\hfill
\end{proof}

\section{Examples}\label{sec:bsp}

We now present three simple examples of sequential recursive algorithms, \textit{mergesort\/}, \textit{quicksort\/} and the \textit{sieve of Eratosthenes\/}. These algorithms will be specified by recursive ASMs, which we use to illustrate the concepts in our axiomatisation. Furthermore, we show that sequential recursion can already be expressed by ASMs, as these support unbounded parallelelism. We illustrate this for the first two selected algorithms. This shows that unbounded parallelism, which is a decisive feature of ASMs, is much stronger than sequential recursion, and there is no need to separately investigate recursive parallel algorithms.

\subsection{Mergesort}
\label{sect:mergesort}

We first give a specification of a recursive ASM comprising two named rules {\em sort} (the main rule) and {\em merge}.

\begin{tabbing}
xxx\=xxxxx\=xxxxx\=xxxxx\=xxxxx\=xxxxx\=xxxxx\= \kill
\> sorted\_list $\leftarrow$ \textit{sort\/}(unsorted\_list) = \\
\>\>\texttt{IF} sorted\_list = \textit{undef\/} \texttt{THEN} \\
\>\> \texttt{LET} $n$ = length(unsorted\_list) \\
\>\> \texttt{IN} \> \texttt{PAR} \\
\>\>\> (\texttt{IF} $n \le 1$ \texttt{THEN} sorted\_list := unsorted\_list ) $\|$ \\
\>\>\> (\texttt{IF} $n > 1 \wedge \text{sorted\_list}_1 = \textit{undef\/} \wedge \text{sorted\_list}_2 = \textit{undef\/}$ \\
\>\>\> \ \texttt{THEN} \> \texttt{LET} $\text{list}_1 = \mathbf{I} L . \text{length}(L) = \lfloor \frac{n}{2} \rfloor \wedge \exists L^\prime . \text{concat}(L,L^\prime) = \text{unsorted\_list}$ \\
\>\>\>\> \texttt{IN} \> \texttt{LET} $\text{list}_2 = \mathbf{I} L^\prime . \text{concat}(\text{list}_1,L^\prime) = \text{unsorted\_list}$ \\
\>\>\>\>\> \texttt{IN} \> \texttt{PAR} \\
\>\>\>\>\>\> (sorted\_list$_1 \leftarrow$ \textit{sort\/}(list$_1$) $\|$ \\
\>\>\>\>\>\> \ sorted\_list$_2 \leftarrow$ \textit{sort\/}(list$_2$))  $\|$ \\
\>\>\> (\texttt{IF} $n > 1 \wedge \text{sorted\_list}_1 \neq \textit{undef\/} \wedge \text{sorted\_list}_2 \neq \textit{undef\/}$ \\
\>\>\> \ \texttt{THEN} \> sorted\_list $\leftarrow$ \textit{merge\/}(\text{sorted\_list}$_1$,\text{sorted\_list}$_2$))
\end{tabbing}

Here we used terms of the form $\mathbf{I} x . \varphi$ with a variable $x$ and a formula $\varphi$, in which $x$ is free to denote {\em the unique value $x$ satisfying $\varphi$}\footnote{In the cases above we could have used equivalently $@ x . \varphi$ denoting an {\em arbitrary value $x$ satisfying $\varphi$}, but emphasising that the value exists and is indeed unique makes the specification clearer. Both kinds of terms were originally introduced by David Hilbert using $\epsilon$ instead of $@$ and $\iota$ instead of $\mathbf{I}$. The notation $@$ reflects the common use of \texttt{ANY} ($@$) in rigorous methods and  $\mathbf{I}$ comes from Fourman's formalisation of higher-order intuitionistic logic.}.

\begin{tabbing}
xxx\=xxxxx\=xxxxx\=xxxxx\=xxxxx\=xxxxx\=xxxxx\= \kill
\> merged\_list $\leftarrow$ \textit{merge\/}(inlist$_1$,inlist$_2$) = \\
\>\>\texttt{IF} merged\_list = \textit{undef\/} \texttt{THEN} \\
\>\> \texttt{PAR} \\
\>\> (\texttt{IF} inlist$_1 = []$ \texttt{THEN} merged\_list := inlist$_2$)  $\|$ \\
\>\> (\texttt{IF} inlist$_2 = []$ \texttt{THEN} merged\_list := inlist$_1$)  $\|$ \\
\>\> (\texttt{IF} $\text{inlist}_1 \neq [] \wedge \text{inlist}_2 \neq []$ \\
\>\> \ \texttt{THEN} \> \texttt{LET} $x_1 = \text{head}(\text{inlist}_1)$ \texttt{IN} \\
\>\>\> \texttt{LET} $\text{restlist}_1 = \text{tail}(\text{inlist}_1)$ \texttt{IN} \\
\>\>\> \texttt{LET} $x_2 = \text{head}(\text{inlist}_2)$ \texttt{IN} \\
\>\>\> \texttt{LET} $\text{restlist}_2 = \text{tail}(\text{inlist}_2)$ \texttt{IN}\\
\>\>\>\> \texttt{PAR} \\
\>\>\>\> (\texttt{IF} $x_1 \le x_2 \wedge \text{merged\_restlist} = \textit{undef\/}$ \\
\>\>\>\> \ \texttt{THEN} merged\_restlist $\leftarrow$ \textit{merge\/}(restlist$_1$, inlist$_2$)) $\|$ \\
\>\>\>\> (\texttt{IF} $x_1 > x_2 \wedge \text{merged\_restlist} = \textit{undef\/}$ \\
\>\>\>\> \ \texttt{THEN} merged\_restlist $\leftarrow$ \textit{merge\/}(inlist$_1$, restlist$_2$)) $\|$ \\
\>\>\>\> (\texttt{IF} $x_1 \le x_2 \wedge \text{merged\_restlist} \neq \textit{undef\/}$ \\
\>\>\>\> \ \texttt{THEN} merged\_list := concat($[x_1]$,merged\_restlist)) $\|$ \\
\>\>\>\> (\texttt{IF} $x_1 > x_2 \wedge \text{merged\_restlist} \neq \textit{undef\/}$ \\
\>\>\>\> \ \texttt{THEN} merged\_list := concat($[x_2]$,merged\_restlist))
\end{tabbing}

The root of the call tree is labelled by \textit{sort\/}. Every node labelled by \textit{sort\/} has three children labelled by instances of \textit{sort\/}, \textit{sort\/} and \textit{merge\/}, respectively. A node labelled by \textit{merge\/} has two children, both labelled by \textit{merge\/}.

Since in a run of a sequential recursive algorithm, in each step only finitely many algorithms are executed at the same time and these do not stand in an ancestor relationship, we can in fact instead of using copies of the algorithms use directly copies of the locations and run all these part-algorithms in parallel. They can only make a step, if their input has been defined. The parallelism is then unbounded, but in each step only finitely many parallel branches are executed. We present such a specification using parallel ASMs for the mergesort algorithm, using sequences of indeces $0,1,2$ for the parameters. The original input we are interested in is $unsorted\_list([~])$, the computed output of interest is $sorted\_list([~])$.\footnote{We use the vertical notation of $r_i$ instead of \texttt{PAR} $r_1 \| \dots \| r_n$.}

\begin{asm}
\FORALL I \in \{0,1,2\}^* \DO \+
\IF sorted\_list(I)= ~ \textit{undef\/} ~ 
     \AND unsorted\_list(I) \not = ~\textit{undef\/} ~ \THEN \+
  \LET n = length(unsorted\_list(I)) \\
  \IF n \leq 1 \THEN  sorted\_list(I):= unsorted\_list(I) \\
  \IF n>1 \THEN \+
     \ASM{SplitList}(I) \\
     \IF SortedSubLists(I) \THEN \+
         \IF SubListsToBeMerged(I)\+
           \THEN ~ \ASM{Merge}(I,sorted\_list(I0),sorted\_list(I1))\\
            \ELSE ~ \ASM{OutputSortingResult}(I)\dec\dec\dec\-     
\end{asm}

The three submachines are defined as follows:
\begin{asm}
	\ASM{SplitList}(I) =\+
   \IF unsorted\_list(I0)  = ~\textit{undef\/} =unsorted\_list(I1) \THEN \+
     \LET (l_0,l_1) \WITH \+
        length(l_0)= \lfloor \frac{n}{2} \rfloor \AND 
           unsorted\_list(I)  = concat(l_0,l_1) \IN \+
               unsorted\_list(I0):=l_0 \\
               unsorted\_list(I1):=l_1\dec \dec\dec\-
\WHERE \+
SortedSubLists(I) \IFF \+
 sorted\_list(I0) \not = ~ \textit{undef\/} ~
     \AND sorted\_list(I1) \not =  ~\textit{undef\/}\-
SubListsToBeMerged(I) \IFF  merged\_list(I)= ~\textit{undef\/} \\
\ASM{OutputSortingResult}(I)=
(sorted\_list(I):=merged\_list(I))
\end{asm}

\begin{asm}
\ASM{Merge}(I,l_1,l_2)=\+
   \IF l_1=[~] \THEN merged\_list(I):=l_2 \\
   \ELSE ~ \IF l_2=[~] \THEN merged\_list(I):=l_1\+
      \ELSE \+
         \ASM{InitializeInputLists}(I,l_1,l_2)\\
         \ASM{InitializeRestLists}(I,l_1,l_2)\\
         \ASM{OneMergeStep}(I) \dec\dec\-
\WHERE \+
\ASM{InitializeInputLists}(I,l_1,l_2)=\+
   \IF inlist_1(I)=inlist_2(I)=~\textit{undef\/} ~\texttt{THEN} \+
       inlist_1(I):=l_1 \\
       inlist_2(I):=l_2 \dec\-
 \ASM{InitializeRestLists}(I,l_1,l_2)=\+
    \IF restlist_1(I)=restlist_2(I)=~\textit{undef\/} ~\THEN \+
       restlist_1(I):=tail(l_1)\\
       restlist_2(I):=tail(l_2) \dec\-
\ASM{OneMergeStep}(I)=\+
   \IF inlist_1(I)\not = [~] \not = inlist_2(I) \AND 
        restlist_1(I) \not =~\textit{undef\/} \not =restlist_2(I)
        \THEN \+
        \LET x_1=head(inlist_1(I)),  x_2=head(inlist_2(I))\+
          \IF x_1 \leq x_2 \THEN \+
             \IF merged\_list(I2)=~\textit{undef\/}\+
                \THEN \+
                   inlist_1(I2):=restlist_1(I)\\
                   inlist_2(I2):=inlist_2(I)\-
                \ELSE merged\_list(I):=concat([x_1],merged\_list(I2))\dec\-
          \IF x_1>x_2 \THEN \+
              \IF merged\_list(I2)=~\textit{undef\/}\+
              \THEN \+
                 inlist_1(I2):=inlist_1(I)\\
                 inlist_2(I2):=restlist_2(I)\-
              \ELSE merged\_list(I):=concat([x_2],merged\_list(I2))           
\end{asm}

\begin{remark*}
Note that in this way every sequential recursive algorithms can be captured by an ASM as defined in \cite{boerger:2003}, i.e. exploiting \texttt{FORALL} rules. This shows that the concept of unbounded parallelism that is decisive for ASMs and is theoretically founded by the behavioural theory of parallel algorithms \cite{ferrarotti:tcs2016} covers the needs of sequential recursive algorithms.\footnote{A reviewer objected that we cannot describe `recursive algorithms with unbounded number of callees'. In fact one cannot do this with sequential recursive ASM because they satisfy the Bounded Exploration and the Bounded Call Tree Branching Postulates. But one can do it using ASMs with unbounded \texttt{FORALL}. Similarly, the reviewer's question how to define `recursive symmetry-braking algorithms' can be answered by using ASMs with unbounded choice.} However, parallel algorithms are a much wider class than sequential recursive algorithms\footnote{This explains why in the behavioural theories developed so far the emphasis was more on parallel than on recursive algorithms. Nonetheless, the ASM rule above expressing the sequential recursive algorithm is somehow easier to read than the rule for the behaviourally equivalent parallel ASM, because in the latter one the copies of locations are made explicit.}. We dispense with discussing this any further here. It further means that the concept of recursion extends the class of sequential algorithms, but an extension of the class of parallel algorithms by means of recursion is meaningless.

Furthermore, instead of using a single parallel ASM with the rule above we can also define agents $a_I$ for all indices $I \in \{ 0,1,2 \}^*$, each associated with an ASM $\mathcal{M}_I$ using a rule $r_I$ defined as above without the outermost \texttt{FORALL}. The effect on runs is that for the concurrent runs of the concurrent ASM $\{ (a_I, \mathcal{M}_I) \mid I \in \{ 0,1,2 \}^* \}$ the individual agents operate asynchronously, which permits additional runs. However, these runs only reflect the combination of runs of the individual agent machines at different paces as already mentioned in the introduction. This will be exploited in Section \ref{sec:poruns}.
\end{remark*}

\subsection{Quicksort}

Again we first give a specification of a sequential recursive ASM with a single named i/o-rule {\em qsort}.

\begin{tabbing}
xxx\=xxxxx\=xxxxx\=xxxxx\=xxxxx\=xxxxx\=xxxxx\= \kill
\> sorted\_list $\leftarrow$ \textit{qsort\/}(unsorted\_list) = \\
\>\> \texttt{PAR} \> (\texttt{IF} unsorted\_list = $[]$ \texttt{THEN} sorted\_list := $[]$) $\|$ \\
\>\>\> (\texttt{IF} $\text{unsorted\_list} \neq [] \wedge \text{sorted\_list}_1 = \textit{undef\/} \wedge \text{sorted\_list}_2 = \textit{undef\/}$ \\
\>\>\> \ \texttt{THEN} \> \texttt{LET} $x = \text{head}(\text{unsorted\_list})$ \texttt{IN} \\
\>\>\>\> \texttt{LET} $\text{unsorted\_list}_1 = \text{sublist}[y \mid y < x](\text{unsorted\_list})$ \texttt{IN} \\
\>\>\>\> \texttt{LET} $\text{unsorted\_list}_2 = \text{sublist}[y \mid y > x](\text{unsorted\_list})$ \texttt{IN} \\
\>\>\>\> \texttt{PAR} \> ($\text{sorted\_list}_1 \leftarrow \textit{qsort\/}(\text{unsorted\_list}_1)$) $\|$ \\
\>\>\>\>\> ($\text{sorted\_list}_2 \leftarrow \textit{qsort\/}(\text{unsorted\_list}_2)$)) \\
\>\>\> (\texttt{IF} $\text{unsorted\_list} \neq [] \wedge \text{sorted\_list}_1 \neq \textit{undef\/} \wedge \text{sorted\_list}_2 \neq \textit{undef\/}$ \\
\>\>\> \ \texttt{THEN} \> $\text{sorted\_list} := \text{concat}(\text{sorted\_list}_1, \text{concat}(\text{head}(\text{unsorted\_list}),\text{sorted\_list}_2))$)
\end{tabbing}

In this case all nodes of the call tree are labelled by instances of the one i/o-algorithm \textit{qsort\/}. The following ASM rule shows how to capture the algorithm by a parallel ASM.

\begin{tabbing}
xxx\=xxxxx\=xxxxx\=xxxxx\=xxxxx\=xxxxx\=xxxxx\= \kill
\> \texttt{FORALL} $I \in \{ 0,1 \}^*$ \texttt{DO} \\
\> \texttt{PAR} \> (\texttt{IF} $\text{unsorted\_list}(I) = [~]$ \texttt{THEN} $\text{sorted\_list}(I) := [~]$) $\|$ \\
\>\> (\texttt{IF} $\text{unsorted\_list}(I) \neq [~] \wedge \text{unsorted\_list}(I) \neq \textit{undef\/}$ \\
\>\> \ \texttt{THEN} \> \texttt{PAR} \\
\>\>\> (\texttt{IF} $\text{unsorted\_list}(I0) = \textit{undef\/} \wedge \text{unsorted\_list}(I1) = \textit{undef\/}$ \\
\>\>\> \ \texttt{THEN} \texttt{LET} $x = \text{head}(\text{unsorted\_list}(I))$ \texttt{IN} \\
\>\>\>\> \texttt{PAR} \> ($\text{unsorted\_list}(I0) := \text{sublist}[y \mid y < x](\text{unsorted\_list}(I))$) $\|$ \\
\>\>\>\>\> ($\text{unsorted\_list}(I1) := \text{sublist}[y \mid y > x](\text{unsorted\_list}(I))$))$\|$ \\
\>\>\> (\texttt{IF} $\text{sorted\_list}(I0) \neq \textit{undef\/} \wedge \text{sorted\_list}(I1) \neq \textit{undef\/}$ \\
\>\>\> \ \texttt{THEN} \> $\text{sorted\_list}(I) := \text{concat}(\text{sorted\_list}(I0),$ \\ 
\>\>\>\>\>\>\> $\text{concat}(\text{head}(\text{unsorted\_list}(I)),\text{sorted\_list}(I1)))$))
\end{tabbing}

As for \textit{mergesort\/}, one can define a concurrent ASM $\{ (a_I, \mathcal{M}_I) \mid I \in \{ 0,1 \}^* \}$, where the rule $r_I$ of the machine $\mathcal{M}_I$ is defined as above without the outermost \texttt{FORALL}.

\subsection{Sieve of Eratosthenes}
\label{sect:sieveEratos}

The \textit{sieve of Eratosthenes\/} algorithm computes all prime numbers. Starting from the set $\{ x \in \mathbb{N} \mid x \ge 2 \}$ as the start sieve, the smallest number of the sieve is added to the output (or printed)---so the output will be infinite---and all elements of the sieve that are divisable by the number added to the output are removed from the sieve. The following is a simple recursive ASM rule for this algorithm.

\begin{tabbing}
xxx\=xxxxx\=xxxxx\=xxxxx\=xxxxx\=xxxxx\=xxxxx\= \kill
\> $\leftarrow$ \textit{eratosthenes\/}(sieve) = \\
\>\> \texttt{LET} \> $p = \min(\text{sieve})$ \texttt{IN} \\
\>\> \texttt{LET} \> $\text{reduced\_sieve} = \{ x \in \text{sieve} \mid p \nmid x \}$ \texttt{IN} \\
\>\> \texttt{PAR} \> $\leftarrow \textit{eratosthenes\/}(\text{reduced\_sieve}) \; \| \; \text{out\_prime} := p$
\end{tabbing}

More generally, if the input sieve is any subset of $\mathbb{N}$, then the algorithm will return all numbers in $\{ x \in \text{sieve} \mid \forall y \in \text{sieve} ~ (y \mid x \Rightarrow y = x) \}$ (via values assigned to out\_prime). None of the calls in this algorithm will reach a final state.

The crucial problem with this algorithm is not only the infinite input and output (both can in principle be handled using streams), but the fact that the involved computation of reduced\_sieve requires the availability of all elements of sieve, i.e. in general infinitely many. A remedy is the following rule, which keeps sieve constant, but uses all previously determined output values in a set divisors (initially $\emptyset$) to determine the next number in the output, which requires only the investigation of an initial segment of sieve.

\begin{tabbing}
xxx\=xxxxx\=xxxxx\=xxxxx\=xxxxx\=xxxxx\=xxxxx\= \kill
\> $\leftarrow$ \textit{eratos\/}(divisors,sieve) = \\
\>\> \texttt{LET} \> $p = \min (\{ x \in \text{sieve} \mid \forall d \in \text{divisors} ~~ d \nmid x \})$ \texttt{IN} \\
\>\> \texttt{PAR} \> $\leftarrow \textit{eratos\/}(\text{divisors} \cup \{ p \}, \text{sieve}) \; \| \; \text{out\_prime} := p$
\end{tabbing}

It is an easy exercise to describe this algorithm by a parallel or a concurrent ASM.

\section{Recursive Algorithms and Partial-Order Runs}\label{sec:poruns}

In their response to Moschovakis' criticism   \cite{moschovakis:2001} Blass and Gurevich argued in  \cite{blass:beatcs2002} with the capture of sequential recursive algorithms by what is known as 
\emph{distributed ASMs}. The semantics of distributed ASMs has been defined by {\em partial-order runs} \cite{gurevich:lipari1995}, which later have been recognized as too restrictive and have been replaced in \cite{boerger:ai2016} by a rather comprehensive definition of concurrency. In this section we investigate the relationship between distributed ASMs and recursive 
algorithms/ASMs, as defined in Section \ref{sec:postulates}. We will show (a precise version of) the following theorem.

\begin{theorem}\label{CharacterizePorun}
{\bf Characterization of partial-order runs.} Recursive algorithms are exactly those finitely-composed concurrent algorithms $\mathcal{C}$ with nd-seq components such that all concurrent $\mathcal{C}$-runs are definable by partial-order runs.
\end{theorem}

The decisive notions used in this theorem such as  partial-order runs and finitely-composed concurrent algorithms will be formally introduced below, together with the corresponding concepts for concurrent ASMs. The theorem fortifies the argument that 
partial-order runs are too weak a concept to serve as a semantic foundation for concurrent algorithms\footnote{As our discussion in the previous section shows, even the concept of unbounded parallelism, by means of which synchronous parallelism is captured, is stronger than recursion.}. The latter aspect has been overcome by the definition of concurrent ASM runs and the corresponding concurrent ASM thesis in \cite{boerger:ai2016}.

We will also show that if the concurrent runs are 
restricted further to partial-order runs of a concurrent algorithm with a fixed finite number of agents and their nd-seq programs, 
one can simulate them even by a non-deterministic sequential 
algorithm (see Theorem \ref{thm-Petri}). An interesting 
example of this special case are partial-order runs of Process Rewrite Systems \cite{Mayr99} (see Corollary \ref{corollary}). 

To prove the theorem we use the characterization of runs of sequential recursive  algorithms as 
recursive runs of recursive ASMs (Theorem 
\ref{thm-capture}) and of runs of concurrent algorithms as 
concurrent ASM runs (see \cite{boerger:ai2016}). 

\subsection{Partial-Order Runs}

Syntactically, a concurrent algorithm $\mathcal{C}$ is 
defined as a family of algorithms $alg(a)$, each associated with (`indexed by') an agent $a \in Agent$ (see \cite{gurevich:lipari1995},  
\cite{boerger:ai2016}) that executes the algorithm in concurrent runs.
These algorithms are often assumed to be (possibly non-deterministic) sequential algorithms, though this 
restriction is not necessary in general. In our case here 
this restriction is, however, important, as we have seen in 
Section \ref{sec:bsp} that without this restriction, permitting unbounded $\FORALL$ and $\CHOOSE$ constructs, we 
obtain algorithms far more powerful than the recursive ones.
Sometimes it is also assumed that the $Agent$ set is finite, 
a special case we consider in Sect.~\ref{sec:Petri}.

In a concurrent run, as in recursive runs, 
different agents can be associated dynamically with different instances 
of the same algorithm. Therefore, when relating concurrent runs of a concurrent algorithm $\mathcal{C}$ to recursive runs of a sequential recursive algorithm---which by Definition \ref{def-recAlg} is a finite set of i/o-algorithms---we need to finitely compose $\mathcal{C}$, namely by a finite set of nd-seq components of which each $alg(a) \in \mathcal{C}$  is an instance.  
 
In a concurrent run as defined in \cite{boerger:ai2016}, 
multiple agents with different clocks may contribute by their single moves to define the successor state of a state.
Therefore, when a successor state $S_{i+1}$ of a state $S_i$ is obtained by applying to $S_i$ multiple update sets $U_a$ with agents $a$ in a finite set $Agent_i \subseteq Agent$, each $U_a$ is required to have been computed by
$a \in Agent_i $ in a preceding state $S_j$, i.e. with $j \leq i$. It is possible that $j<i$ holds so that for different agents different $\text{alg}(a)$-execution speeds (and purely local subruns to compute $U_a$) can be taken into account. 

This can be considered as resulting from a separation of a step of an nd-seq algorithm $\text{alg}(a)$ into a {\em read step}---which reads location values in a state $S_j$---followed by a {\em write step} which applies the update set $U_a$ computed on the basis of the values read in $S_j$ to a later state $S_i$ ($i \geq j$). We say that $a$ contributes to updating the state $S_i$ to the successor state $S_{i+1}$, and that a $move$ starts in $S_j$ and contributes to updating $S_i$ (i.e. it finishes in $S_{i+1}$). This is formally expressed by the following definition of concurrent ASMs and their runs.

\begin{definition}\label{def-concurRun}\rm
	
A {\em concurrent ASM} is defined as a family $\mathcal{C}$ of ASMs $\text{asm}_a$ (also called programs and written $pgm(a)$) with associated agents $a \in Agent$. A {\em concurrent run} of $\mathcal{C}$ is defined as a sequence $S_0, S_1, \dots$ of states together with a sequence $A_0, A_1, \dots$ of finite subsets of $Agent$, such that $S_0$ is an initial state and each $S_{i+1}$ is obtained from $S_i$ by applying to it the updates computed by the agents in $A_i$, where each $a \in A_i$ computes its update  set  $U_a$ on the basis of the location values (including the input and shared locations) read in some preceding state $S_j$ (i.e. with $j \leq i$) depending on $a$.
	
\end{definition} 

\begin{remark*}
In this definition we deliberately permit the
set of $Agent$s to be infinite or dynamic and potentially infinite, 
growing or shrinking in a run. Below we consider the special 
cases that $Agent$ is a static finite set (see Section 
\ref{sec:Petri}) or a dynamic set all of whose members are equipped however
with instances of a fixed (static) finite set of programs (see 
Definition \ref{def-composed-casm}). In Definition \ref{def-porun} below the set of $Agent$s is fixed by the set of $M$oves.
\end{remark*}

For the reason explained above, in the following we restrict our attention to concurrent ASMs in which each component $\text{asm}_a$ is an nd-seq ASM.
 
In \cite{gurevich:lipari1995} Gurevich defined the notion of 
{\em partial-order run} of concurrent 
algorithms.\footnote{Note that in \cite{gurevich:lipari1995} 
	Gurevich actually uses the wording {\em distributed algorithm} instead of concurrent algorithm, whereas we prefer to stick with the notation from \cite{boerger:ai2016}. One reason for this is that distribution requires also to discuss (physical) notations and messaging (as handled in \cite{boerger:jucs2017}), whereas concurrency abstracts from this.} 
The partial order is defined on the set of single moves of the agents which execute the individual algorithms. 
For a nd-seq algorithm $\mathcal{A}$, to make one {\em move}
means to perform one step in a state $S$, as defined by the 
Branching Time Postulate \ref{p-time}, applying 
a set $\Delta(S,S^\prime) \in \boldsymbol{\Delta}(S)$ 
of updates to $S$.

\begin{definition}\label{def-porun}\rm
	
Let $\mathcal{C} = \{ (a, \text{alg}(a)) \}_{a \in Agent}$ be a concurrent algorithm, in which each $\text{alg}(a)$ is an nd-seq algorithm. A {\em partial-order run} for $\mathcal{C}$ is defined by a set $M$ of moves of instances of the algorithms $\text{alg}(a)$ ($a \in Agent$), a function $\text{ag}: M \rightarrow Agent$ assigning to each move the agent performing the move, a partial order $\le$ on $M$, and an initial segment function $\sigma$ such that the following conditions are satisfied:

\begin{description}
		
\item[finite history.] For each move $m \in M$ its history $\{ m^\prime \mid m^\prime \le m \}$ is finite.
		
\item[sequentiality of agents.] The moves of each agent are ordered, i.e. for any two moves $m$ and $m^\prime$ of one agent $\text{ag}(m) = \text{ag}(m^\prime)$ we either have $m \le m^\prime$ or $m^\prime \le m$.
		
\item[coherence.] For each finite initial segment $M^\prime \subseteq M$ (i.e. for $m \in M^\prime$ and $m^\prime \le m$ we also have $m^\prime \in M^\prime$) there exists a state $\sigma(M^\prime)$ over the combined signatures of the algorithms $\mathcal{A}_i$ such that for each maximum element $m \in M^\prime$ the state $\sigma(M^\prime)$ is the result of applying $m$ to $\sigma(M^\prime - \{ m \})$.	
			
\end{description}
	
\end{definition}

In order to characterise recursive ASM runs in terms of partial-order runs of a concurrent ASM $\mathcal{C}$ several restrictions have to be made. First of all the component ASMs must be nd-seq ASMs, an assumption we already justified above. Second, the component machines $\text{asm}_a$ can only be copies (read: instances) of finitely many different ASMs, which we will call the program base of 
$\mathcal{C}$.\footnote{This reflects the stipulation in 
	\cite{gurevich:lipari1995} that in concurrent ASMs the agents are equipped with instances of programs which are taken from `a finite indexed set of single-agent programs' (p.31).} 
Third, runs must be started by executing a distinguished 
main component.\footnote{The restriction to one component is equivalent to, but notationally simplifies, the requirement stated in \cite[6.2, p.31]{gurevich:lipari1995} for concurrent ASM runs 
	that in initial states there are only finitely many agents, each equipped with a program.} 
We capture these restrictions by the notion of {\em finitely composed} concurrent ASM.

\begin{definition}\label{def-composed-casm}\rm

A concurrent ASM $\mathcal{C}$ is {\em finitely composed} iff (i) and (ii) hold:
	
\begin{enumerate}
		
\item There exists a finite set $\mathcal{B}$ of nd-seq ASMs such that each $\mathcal{C}$-program is of form $\textbf{amb}\; a \;\textbf{in}\; r$ for some program $r \in \mathcal{B}$---call $\mathcal{B}$ the {\em program base} of $\mathcal{C}$.
		
\item There exists a distinguished agent $a_0$ which is the only one $Active$ in any initial state. Formally this means that in every initial state  of a $\mathcal{C}$-run, $Agent=\{a_0\}$ holds. We denote by \textit{main\/} the component in $\mathcal{B}$ of which $a_0$ executes an instance. For partial-order runs of $\mathcal{C}$ this implies that they start with a minimal move which consists in executing the program $\text{asm}(a_0) =~ \AMB a_0 \IN \textit{main\/}$.
		
\item Any program base component may contain rules of form $\LET a =~\NEW(Agent) \IN r$.\footnote{This reflects the stipulation for concurrent ASMs in \cite{gurevich:lipari1995} that `An agent $a$ can \emph{make a move} at $S$ by firing Prog($a$) ... and change $S$ accordingly. As part of the move, $a$ may create new agents' (p.32), which then may contribute by their moves to the run in which they were created.} Together with (ii) this implies that every agent, except the distinguished $a_0$, before making a move in a run must have been created in the run. 

\end{enumerate}
	 	 
\noindent
$\mathcal{C}$ is called {\em finite} iff $Agent$ is finite.

\end{definition}

We are now ready to more precisely state and prove the first part of Theorem \ref{CharacterizePorun}. It should come as no surprise; it provides the justification for the argumentation by Blass and Gurevich in \cite{blass:beatcs2002}.

\begin{theorem}\label{thm-porun}

For every recursive ASM $\mathcal{R}$ one can construct an equivalent finitely composed concurrent ASM $\mathcal{C}_{\mathcal{R}}$ with nd-seq ASM components such that every concurrent run of $\mathcal{C}_{\mathcal{R}}$  is definable by a partial-order run.

\end{theorem}

\begin{proof}
Let $\mathcal{R}$ be a recursive ASM given with distinguished program $main$. We define a finitely composed concurrent ASM $\mathcal{C}_{\mathcal{R}}$ with program base $\{r^* \mid r \in \mathcal{R}\}$, where $r^*$ is defined as
\begin{asm}
	r^* = ~ \IF Active(r) \AND ~\NOT Waiting(r) \THEN r.	
\end{asm}

That is, $r^*$ can only contribute a non-empty update set to form a state $S_{i+1}$ in a concurrent run, if $r$ is $Active$ and not $Waiting$; this is needed, because in every step of a recursive run of $\mathcal{R}$ only $Active$ and not $Waiting$ rules are executed.

In doing so we use for each call rule $r \in \mathcal{R}$ in the $\THEN$ part of $r^*$ instead of
$t_0 \leftarrow N(t_1,\dots,t_n)$ its interpretation
by the ASM rule $\ASM{Call}(t_0 \leftarrow N(t_1,\dots,t_n))$ defined in Sect.~\ref{sect:recASM}.
The definition of $r^*$ obviously guarantees the behavioral equivalence of $\mathcal{R}$ and $\mathcal{C}_{\mathcal{R}}$: in each run step the same $Active$ and not $Waiting$ rules $r$ respectively $r^*$ and their agents are selected for their simultaneous  execution.
Remember that, by the definition of $\ASM{Call}$(\textit{i/o-rule\/}),
each agent operates in its own state space so that the view of an agent's step as read-step followed by a write-step is equivalent to the atomic view of this step.

Note that in a concurrent run of $\mathcal{C}_{\mathcal{R}}$ the $Agent$ set is dynamic, in fact it grows with each execution of a call rule, together with the number of instances of $\mathcal{R}$-components executed during a recursive run of $\mathcal{R}$.

It remains to define every concurrent run $(S_0,A_0),(S_1,A_1),\ldots$ of $\mathcal{C}_{\mathcal{R}}$ by a partial-order run. For this we define an order on the set $M$ of moves made during a concurrent run, showing that it satisfies the constraint on finite history and sequentiality of agents, and then relate each state $S_i$ of the run to the state computed by the set $M_i$ of moves performed to compute $S_i$ (from $S_0$), showing that $M_i$ is a finite initial segment of $M$ and that the associated state $\sigma(M_i)$ equals $S_i$ and satisfies the coherence condition.

Each successor state $S_{i+1}$ in a concurrent run of $\mathcal{C}_{\mathcal{R}}$ is the result of applying to $S_i$ the write steps of finitely many moves of agents in $A_i$. This defines the function $ag$, which associates agents with moves, and the finite set $M_i$ of all moves finished in a state belonging to the initial run segment $[S_0,\ldots,S_i]$. Let $M=\cup_i M_i$. The partial order $\le$ on $M$ is defined by $m < m^\prime$ iff move $m$ contributes to update some state $S_i$ (read: finishes in $S_i$) and move $m^\prime$ starts reading in a later state $S_j$ with $i+1 \le j$. Thus, by definition, $M_i$ is an initial segment of $M$. 

To prove the finite history condition, consider any $m^\prime \in M$ and let $S_j$ be the state in which it is started. There are only finitely many earlier states $S_0 ,\dots, S_{j-1}$, and in each of them only finitely many moves $m$ can be finished, contributing to update $S_{j-1}$ or an earlier state. 

The condition on the sequentiality of the agents follows directly from the definition of the order relation $\le$ and from the fact that in a concurrent run, for every move $m=(read_m,write_m)$ executed by an agent, this agent performs no other move between the $read_m$-step and the corresponding $write_m$-step in the run.

This leaves us to define the function $\sigma$ for finite initial segments $M^\prime \subseteq M$ and to show the coherence property. 
We define $\sigma(M^\prime)$ as result of the application of the moves in $M^\prime$ in any total order extending the partial order $\le$. For the initial state $S_0$ we have $\sigma(\emptyset)=S_0$. This implies the definability claim $S_i=\sigma(M_i)$.

The definition of $\sigma$ is consistent for the following reason. Whenever two moves $m \not = m^\prime$ are incomparable, then either they both start in the same state or say $m$ starts earlier than $m^\prime$.
But $m^\prime$ also starts earlier than $m$ finishes. This is only possible for agents $ag(m) = a$ and $ag(m^\prime) = a^\prime$ whose programs $pgm(a),pgm(a^\prime)$ are not in an ancestor relationship in the call tree. Therefore these programs have disjoint signatures, so that the moves $m$ and $m^\prime$ could be applied in any order with the same resulting state change.

To prove the coherence property let $M^\prime$ be a finite initial segment, and let $M^{\prime\prime} = M^\prime \setminus M^\prime_{\text{max}}$, where $M^\prime_{\text{max}}$ is the set of all maximal elements of $M^\prime$. Then $\sigma(M^\prime)$ is the result of applying simultaneously all moves $m \in M^\prime_{\text{max}}$ to $\sigma(M^{\prime\prime})$, and the order, in which the maximum moves are applied is irrelevant. This implies in particular the desired coherence property.\hfill
\end{proof}

Note that the key argument in the proof exploits the fact that for recursive runs of $\mathcal{R}$, the runs of  different agents are initiated by calls and concern different state spaces with pairwise disjoint signatures, due to the function parameterization by agents, unless $pgm(a^\prime)$ is a child (or a descendant) of $pgm(a)$, in which case the relationship between the signatures is defined by the call relationship. Independent moves can be guaranteed in full generality only for algorithms with disjoint signatures.

\subsection{Capture of Partial-Order Runs}

While Theorem \ref{thm-porun} is not surprising, we will now show the less obvious converse of Theorem \ref{CharacterizePorun}. 
The fact that (by Definition \ref{def-rASM}) a recursive ASM $\mathcal{R}$ is a finite set of recursive ASM rules, starts its runs with a main program and uses during any run only instances of its rules implies that if $\mathcal{R}$ simulates a concurrent ASM, this concurrent ASM must be finitely composed (as assumed in 
\cite[Sect. 6]{gurevich:lipari1995} for the definition of partial-order runs) and must use only instances of its finitely many nd-seq components. 

\begin{theorem}\label{thm-porun'}
	
For each finitely composed concurrent ASM $\mathcal{C}$ 
with program base $ \{ r_i \mid i \in I \}$ of nd-seq 
ASMs such that all its concurrent runs are definable by 
partial-order runs, one can construct a recursive ASM 
$\mathcal{R}_\mathcal{C}$ which simulates $\mathcal{C}$, i.e. such that for each concurrent run 
of $\mathcal{C}$ there is an equivalent recursive run of 	
$\mathcal{R}_\mathcal{C}$.\footnote{Equivalence 
	via an inverse simulation of every recursive 
		$\mathcal{R}_\mathcal{C}$-run by a concurrent 
		$\mathcal{C}$-run can be obtained if the delegates of $\mathcal{C}$-agents, called in the  
		recursive run to perform the step of their $caller$ in the concurrent run, act in an `eager' way. See the remark at the end of the 
		proof.}
	
\end{theorem}

\begin{proof}
Let a concurrent $\mathcal{C}$-run $(S_0,A_0),(S_1,A_1),\ldots$ be given. If it is definable by a partial-order 
run $(M,\leq,ag,pgm,\sigma)$, the transition from $S_i = 
\sigma(M_i)$ to $S_{i+1}$ is performed in one concurrent 
step by parallel independent moves 
$m \in M_{i+1} \setminus M_i$, where $M_i$ is the set of 
moves which contributed to transform $S_0$ into $S_i$. Let $m \in M_{i+1} \setminus M_i$ be a move performed by an agent $a=ag(m)$ with rule 
$pgm(a)=~ \AMB a \IN r$, an instance of a rule 
$r$ in the program base of $\mathcal{C}$. To 
execute the concurrent step by steps of a 
recursive ASM $\mathcal{R}_\mathcal{C}$, we simulate each 
of its moves $m$ by letting agent $a$ act in the $\mathcal{R}_\mathcal{C}$-run as $caller$ of a named 
rule $out_r \leftarrow \ASM{OneStep}_r(in_r)$. The callee agent $c$ acts as delegate for one step of $a$: it executes $\AMB a \in r$ and makes its program immediately $Terminated$.

To achieve this, we refine the recursive $\ASM{Call}$ machine of Definition \ref{def-CallNamedRule}
to a recursive ASM $\ASM{Call}^*$ by adding to $\ASM{Initialize}$ the
update $Terminated(\AMB c \IN q):=false$. When $\ASM{Call}^*$ is applied to $out_r \leftarrow \ASM{OneStep}_r(in_r)$, the update of  $Terminated$ makes the delegate $Active$ so that it can make a step to execute 
$\AMB c \IN \ASM{OneStep}_r$. $\ASM{OneStep}_r$ is defined to perform $ \AMB caller(c) \IN r$ and immediately terminate (by setting 
$Terminated$ to true). For ease of exposition we add in Definition \ref{def-CallNamedRule} also the update $caller(c):=\SELF$, to distinguish agents in the concurrent run---the $caller$s of $\ASM{OneStep}_r$-machines---from the delegates each of which simulates one step of its $caller$ and immediately terminates its life cycle.

It remains to determine the input and output for calling $\ASM{OneStep}_r$. For the input we exploit the existence of a bounded exploration witness $W_r$ for $r$. All updates produced in a single step are determined by the values of $W_r$ in the state, in which the call is launched. So $W_r$ defines the input terms of the called rule $\ASM{OneStep}_r$, combined in $in_r$. Analogously, a single step of $r$ provides updates to finitely many locations that are determined by terms appearing in the rule, which defines $out_r$.

We summarize the explanations by the following definition: 
\begin{asm}
\mathcal{R}_\mathcal{C}= \{out_r \leftarrow 
\ASM{OneStep}_r(in_r) \mid r \in \mbox{ program base of }\mathcal{C}\}\\
\ASM{OneStep}_r =\+
   \AMB caller(\SELF) \IN r 
   \mbox{ // the delegate executes the step of its }caller\\
   Terminated(pgm(\SELF)):=true 
               \mbox{ // ... and immediates stops}
\end{asm}

Note that each caller step step $\AMB a \IN ~ out_r \leftarrow \ASM{OneStep}_r(in_r)$ in an $\mathcal{R}_\mathcal{C}$-run is by definition equivalent to the machine $\AMB a \IN ~\ASM{Call}^*(out_r \leftarrow \ASM{OneStep}_r(in_r))$ and triggers the execution of the delegate program $\AMB c \IN \ASM{OneStep}_r$ (where $a=caller(c)$, which triggers 
$\AMB c \IN ~\AMB a \IN r $ (by definition). Furthermore, since the innermost ambient binding counts, this machine is equivalent to the simulated $\mathcal{C}$-run step $\AMB a \IN r$.

Thus the recursive $\mathcal{R}_\mathcal{C}$-run which simulates $(S_0,A_0),(S_1,A_1),\ldots$ starts by Definition \ref{def-composed-casm} in $S_0$\footnote{For the sake of notational simplicity we disregard the auxiliary locations of $\mathcal{R}_\mathcal{C}$.} with program
$\AMB a_0 \IN out_{main} \leftarrow \ASM{OneStep}_{main}(in_{main})$. 
Let  
\begin{asm}
A_i= \{ a_{i_1},\ldots, a_{i_k} \} \subseteq Agent
 \mbox{  for some $i_j$ and  $k$ depending on $i$}\\
\WHERE ~ 
    a_{i_j}=ag(m_{i_j}) \in M_{i+1} \setminus M_i \AND 
   pgm(a_{i_j})=~\AMB a_{i_j} \IN r_{i_j} ~ 
       (\FORALL 1 \leq j \leq k)
\end{asm} 

We use the same agents $a_{i_j}$ for $A_i$ in the $\mathcal{R}_\mathcal{C}$-run, but with  $out_{r_{i_j}}\leftarrow 
\ASM{OneStep}_{r_{i_j}}(in_{r_{i_j}})$ as program. Their step in the recursive run leads 
to a state $S_i' $ where all callers $a_{i_j}$ are $Waiting$ 
and the newly created delegates $c_{i_j}$ are $Active$ and not 
$Waiting$. So we can choose them for the set $A_i'$ of agents which perform  
the next $\mathcal{R}_\mathcal{C}$ step, whereby 
\begin{itemize}
	\item all rules $r_{i_j}$ are performed simultaneously (as in the given concurrent run step), in the ambient of $caller(c_{i_j})=a_{i_j}$ thus leading as desired to the state $S_{i+1}$,
	\item  the delegates make their program $Terminated$, whereby their callers $a_{i_j}$ become again not $Waiting$ and thereby ready to take part in the next step of the concurrent run. We assume for this that whenever in the $\mathcal{C}$-run (not in the $\mathcal{R}_\mathcal{C}$ run) a new agent $a$ is created, it is made not $Waiting$ (by initializing $CalledBy(a):=\emptyset$).
\end{itemize}\hfill
\end{proof}

\begin{remark*}
Consider an $\mathcal{R}_\mathcal{C}$-run where each recursive step of the concurrent caller agents in $A_i$, which call each some $\ASM{OneStep}$ program, alternates with a recursive step of all---the just called---delegates whose program is not yet $Terminated$. Then this run is equivalent to a corresponding concurrent $\mathcal{C}$-run. 
\end{remark*} 

Note that Theorem \ref{thm-porun'} heavily depends on the prerequisite that $\mathcal{C}$ only has partial-order runs. With general concurrent runs as defined in \cite{boerger:ai2016} the construction would not be possible.\footnote{The other prerequisites in Theorem \ref{thm-porun'} appear to be rather natural. Unbounded runs can only result, if in a single step arbitrarily many new agents are created. Also, infinitely many different rules associated with the agents are only possible, if new agents are created and added during a concurrent run. Though this is captured in the general theory of concurrency in \cite{boerger:ai2016}, it was not intended in Gurevich's definition of partial-order runs.} The rather strong conclusion from this is that the class of sequential recursive algorithms is already rather powerful, as it captures all attempts to capture asynchronous parallelism by a formalism that includes some form of a lockstep application of updates defined by several agents. However, true asynchronous behaviour only results, if it is possible that steps by different agents can be started as well as terminated in an independent way, which is covered by the theory of concurrency in \cite{boerger:ai2016}.

\subsection{Finite Static Concurrent ASMs with partial-order runs}\label{sec:Petri}

In this section we consider the special case of static finite concurrent ASMs, which means by Definition \ref{def-concurRun} a static set of pairs $(a,pgm(a))$. These ASMs have fixed finite sets of agents and programs and a fixed association of each program with an executing agent, so that there is no rule instantiation with new agents which could be created during a run. Therefore one can define global states as the union of the component states and the functions $\sigma(I)$ associated with the po-runs yield for every finite initial 
segment $I$ as value the global state obtained by firing the rules in $I$. 

For this particular kind of concurrent ASMs with partial-order runs one can define the concurrent runs by runs of nd-seq ASMs, as we are going to show in this section. This theorem illustrates the rather special character the axiomatic coherence condition imposes on partial-order runs.

\begin{theorem}\label{thm-Petri}
	
For each finite static concurrent ASM $\mathcal{C} = \{ (a_j,r_j) \mid 1 \leq j  \leq n \}$ with nd-seq component ASMs $r_j$ such that all its concurrent runs are definable by partial-order runs one can construct a nd-seq ASM $\mathcal{M}_\mathcal{C}$ such that the concurrent runs of $\mathcal{C}$ and the runs 
of $\mathcal{M}_\mathcal{C}$ are equivalent.
	
\end{theorem}

\begin{corollary}\label{corollary}
The partial-order runs of every Process Rewrite System \cite{Mayr99} 
can be simulated by runs of a non-deterministic sequential ASM.
\end{corollary}

\begin{proof}
For the states $S_i$ of a given concurrent run of $\mathcal{C} $ let $\sigma(M_i)$ be  
the state associated with an initial segment 
$M_i$ of a corresponding partial order run $(M,\leq,ag,pgm,\sigma)$, where each step leading from 
$S_i$ to $S_{i+1}$ consists of pairwise incomparable moves 
in $M_{i+1} \setminus M_i$.  We call such a sequence $S_0, S_1, \dots$ of states a {\em linearised run} of $\mathcal{C}$. For $i>0$ the initial segments 
$M_i$ are non empty.

The linearized runs of $\mathcal{C}$ can be characterized as 
runs of a nd-seq ASM $\mathcal{M}_{\mathcal{C}}$: in each 
step this machine chooses one of finitely many non-empty subsets of the fixed finite set of rules in 
$\mathcal{C}$ to execute them in parallel. Formally:
\begin{asm}
	\mathcal{M}_{\mathcal{C}}= ~
	\CHOOSE 
	\ASM{AllRulesOf}(J_1) \mid \dots \mid          
	\ASM{AllRulesOf}(J_m)\\
	\WHERE \+
	\ASM{AllRulesOf}(\{j_1 ,\dots, j_k\}) = \+
	r_{j_1}\\
	\ldots \\
	r_{j_k} \-
	\{J_1 ,\dots, J_m\}=\{J^\prime \not = \emptyset \mid 
	J^\prime \subseteq J \}	\mbox{  //  the non-empty subsets of } J\\	
	J= \{i  \mid 1 \leq j  \leq n \} \mbox{  //  the fixed set of rule indices} \\
	m = 2^{|J|} - 1 
\end{asm}

To complete the proof it suffices to show the following lemma.\hfill
\end{proof}

\begin{lemma}\label{lem-linearrun}
	
The linearised runs of $\mathcal{C}$ are exactly the runs of $\mathcal{M}_{\mathcal{C}}$.
	
\end{lemma}

\begin{proof}
To show that each run $S_0, S_1, \dots$ of the nd-seq ASM $\mathcal{M}_{\mathcal{C}}$ is a linearised run of $\mathcal{C}$ we proceed by induction to construct the underlying partial-order run $(M,\le)$ with its finite initial segments $M_i$. For the initial state $S_0=\sigma(\emptyset)$ there is nothing to show, so let $S_{i+1}$ result from $S_i$ by applying an update set produced by $\ASM{AllRulesOf}(J^\prime)$ for some non-empty set $J^\prime =\{j_1 ,\dots, j_k\} \subseteq J$. By induction we have $S_i = \sigma(M_i)$ for some initial segment of a partial-order run $(M,\le)$. Since $\ASM{AllRulesOf}(J^\prime)$ is a parallel composition, $S_{i+1}$ results from applying the union of update sets $\Delta_{j_l} \in \boldsymbol{\Delta}_{r_{j_l}}$ for $l = 1,\dots, k$ to $S_i$. Each $\Delta_{j_l}$ defines a move $m_{j_l}$ of some $\text{ag}(m_{j_l}) = a_{j_l}$, move which finishes in state $S_i$. We now have two cases: 
	
\begin{enumerate}
		
\item The moves $m_{j_l}$ with $1 \le l \le k$ are pairwise independent, i.e. their application in any order produces the same new state. Then $(M,\le)$ can be extended with these moves such that $M_{i+1} = M_i \cup \{ m_{j_1}, \ldots, m_{j_k}\}$ becomes an initial segment and $S_{i+1} = \sigma(M_i)$ holds.
		
\item If the moves $m_{j_l}$ with $1 \le l \le k$ are not pairwise independent, the union of the corresponding update sets is inconsistent, hence the run terminates in state $S_i$.
		
\end{enumerate}
	
To show the converse we proceed analogously. If we have a linearized run of states $S_i = \sigma(M_i)$ for all $i \ge 1$, then $S_{i+1}$ results from $S_i$ by applying in parallel all moves in $M_{i+1} \setminus M_i$. Applying a move $m$ means to apply an update set produced by some rule $r_j$ of $\mathcal{C}$ in state $S_i$, and applying several update sets in parallel means to apply their union $\Delta$, which then must be consistent. So we have $S_{i+1} = S_i + \Delta$ with $\Delta = \bigcup_{j \in J^\prime} \Delta_{j}$ for some non-empty $J^\prime =\{j_1 ,\dots, j_k\} \subseteq J$, where each $\Delta_{j_l}$ is an update set produced by $r_{j_l}$ (for $1 \le l \le k$), i.e. $\Delta$ is an update set produced by $\ASM{AllRulesOf}(J^\prime)$, which implies that the linearised run $S_0, S_1, \dots$ is a run of $\mathcal{M}_{\mathcal{C}}$.\hfill	
\end{proof} 

\section{Related Work}\label{sec:more}

The behavioural theory developed in this article contributes to answer the fundamental epistemological question ``What is an algorithm?''. It has been inspired by Gurevich's behavioural theory of sequential algorithms \cite{gurevich:tocl2000}, the ur-instance of a behavioural theory, and motivated by Moschovakis' claim that recursive algorithms, which obviously cannot be modeled `closely and faithfully' by sequential ASMs, can be `directly expressed' by systems of recursive equations (called `recursive programs') \cite[p.100]{Moschovakis19}. 

In \cite[p.99]{Moschovakis19} Moschovakis places ``the basic foundational problem
of \emph{defining algorithms}... outside the scope of this book" and treats \emph{recursive algorithms} as ``faithfully expressed" (ibid. p.101) by syntactically well defined \emph{recursive programs} (read: systems of recursive equations) which permit to compute (partial) functions from auxiliary (in ASM terminology background) functions in whatever given structures. 
Besides restricting the attention to algorithms which compute (partial) functions as least fixed point of a system of equations, such a \emph {definition} is fundamentally different from the behavioural theory approach to \emph{capture} recursive algorithms by a class of abstract machines which can be shown to satisfy an a priori given precise, axiomatic, programming language independent characterization of recursion. 
Furthermore, we use a class of machines which provide a general framework to characterize besides sequential or recursive algorithms also other classes of algorithms, e.g. parallel, interactive or reflective algorithms (see below), which are deliberately left out in \cite{Moschovakis19}. 

Nevertheless, every function which is computable by a recursive program in the sense of Moschovakis can be computed (in the standard meaning of the term) by a recursive ASM. This can be shown easily, for example by using Moschovakis' `recursive machine' \cite[Sect.2D]{Moschovakis19}, an abstract machine which is considered by its author as ``one of the classical implementations of recursion" (ibid.p.74).\footnote{As pointed out in \cite{blass:beatcs2002} and also in \cite{boerger:asm2003}, using the set of equations of a recursive program to compute a concrete function value still requires a determination of control, i.e. in which way the recursion equations are to be applied, a feature which is considered in \cite{Moschovakis19} as implementation detail.} Alternatively one can use sequential recursive ASMs to describe the fixed point construction for systems of recursive equations. Apparently it can also be shown that vice versa, every function which is computable by a recursive ASM can be computed by a recursive program in the sense of Moschovakis (use an induction on the recursion depth of recursive ASMs, with nd-seq ASMs at the basis of the induction). 

As the sequential ASM thesis shows, the notion of sequential algorithm includes a form of bounded parallelism, which is a priori defined by the algorithm and does not depend on the actual state.\footnote{Note that by their definition, Moschovakis' recursive programs satisfy the bounded exploration postulate and their non-deterministic version \cite[Sect. 2E]{Moschovakis19} is carefully restricted to bounded choice.} However, parallel algorithms, e.g. for graph inversion or leader election, require unbounded parallelism. A behavioural theory of synchronous parallel algorithms has been first approached by Blass and Gurevich \cite{blass:tocl2003,blass:tocl2008}, but different from the sequential thesis the theory was not accepted, not even by the ASM community despite its inherent proof that ASMs \cite{boerger:2003} capture parallel algorithms. One reason is that the axiomatic definition exploits non-logical concepts such as mailbox, display and ken, whereas the sequential thesis only used logical concepts such as structures and sets of terms\footnote{Even the background, that is left implicit in the sequential thesis, only refers to truth values and operations on them.}. 

In \cite{ferrarotti:tcs2016} an alternative behavioural theory of synchronous parallel algorithms (aka ``simplified parallel ASM thesis'') was developed. It was inspired by previous research on a behavioural theory for non-deterministic database transformations \cite{schewe:ac2010}. Largely following the careful motivation in \cite{blass:tocl2003} it was first conjectured in \cite{schewe:abz2012} that it should be sufficient to generalise bounded exploration witnesses to sets of multiset comprehension terms\footnote{The rationale behind this conjecture is that in a particular state the multiset comprehension terms give rise to multisets, and selecting one value out each of these multisets defines the proclets used by Blass and Gurevich.} and to make assumptions about background domains, constructors and operations for truth values, records and finite multisets explicit\footnote{The latter aspect was already part of the thesis by Blass and Gurevich.}. The formal proof of the simplified ASM thesis in \cite{ferrarotti:tcs2016} requires among others an investigation in finite model theory.

At the same time another behavioural theory of parallel algorithms was developed in \cite{dershowitz:igpl2016}, which is independent from the simplified parallel ASM thesis\footnote{Apparently, authors of \cite{ferrarotti:tcs2016} and \cite{dershowitz:igpl2016} seemed not to be aware of each others' research.}, but refers also to previous work by Blass and Gurevich. It is debatable, whether the criticism of the defining postulates by Blass and Gurevich also applies to this work; a thorough comparison with the simplified parallel ASM thesis has not yet been conducted.

There have been many attempts to capture asynchronous parallelism, as marked in theories of concurrency as well as distribution (see \cite{lynch:1996} for a collection of many distributed or concurrent algorithms). Commonly known approaches are among others the actor model \cite{agha:1986}, Process Algebras \cite{HandbookPa}, Petri nets \cite{best:1996}, high-level Petri nets \cite{genrich:tcs1981}, and trace theory \cite{mazurkiewicz:lncs1987}. Gurevich's axiomatic definition of partial-order runs \cite{gurevich:lipari1995} tries to reduce the problem to families of sequential algorithms, but the theory is too strict. As shown in \cite{boerger:ai2016} it is easy to find concurrent algorithms that satisfy sequential consistency \cite{lamport:tc1979}, but their runs are not partial-order runs. One problem is that the requirements for partial-order runs always lead to linearisability.

The lack of a convincing definition of asynchronous parallel algorithms was overcome by the work on concurrent algorithms in \cite{boerger:ai2016}, in which a concurrent algorithm is defined by a family of agents, each equipped with a sequential algorithm, possibly with shared locations. While each individual sequential algorithm in the family is defined by the postulates for sequential algorithms\footnote{A remark in \cite{boerger:ai2016} states that the restriction to sequential algorithms is not really needed. An extension to concurrent algorithms covering families of parallel algorithms is handled in \cite{schewe:acsw2017}.}, the family as a whole is subject to a concurrency postulate requiring that in a concurrent run, a successor state of the global state of the concurrent algorithm results from simultaneously applying update sets of finitely many agents that have been built on some previous states (not necessarily the current one). The theory shows that concurrent algorithms are captured by concurrent ASMs. Given the fact that in concurrent algorithms, in particular in case of distribution, message passing between agents is more common than shared locations, it has further been shown in \cite{boerger:jucs2017} that message passing can be captured by regarding mailboxes as shared locations, which leads to communicating concurrent ASMs capturing concurrent algorithms with message passing. In \cite{ferrarotti:scp2019} it has been shown how the popular bulk synchronous parallel bridging model can be captured by a specialised behavioural theory that builds on top of the concurrent ASM thesis in \cite{boerger:ai2016}.

Recently, there is an increased interest in distributed adaptive systems. Adaptivity refers to the ability of an algorithm to modify itself, which is known as linguistic reflection. A behavioural theory of reflective sequential algorithms has been developed in \cite{schewe:tcs2019}.\footnote{A preliminary version of this theory appeared in \cite{ferrarotti:psi2017}.} Again the key aspect is the generalisation of bounded exploration witnesses, which for reflective algorithms comprise terms that can be evaluated to terms and these to values in the base set, so coincidence after double evaluation is required for the equality of update sets in states. The integration of the behavioural theories for parallelism, concurrency and reflection has been sketched in \cite{schewe:acsw2017}, but a more detailed presentation of the combined theory still has to be written up.

\section{Conclusion}\label{sec:schluss}
The main contribution of this article is a behavioural theory of sequential recursive algorithms, providing a) a purely logical definition of this notion, which is independent from any particular abstract machine or programming model, b) a natural extension of nd-seq ASMs to recursive ASMs, and c) a proof that recursive ASMs capture sequential recursive algorithms. The resulting recursive ASM thesis shows (together with the sequential and the concurrent ASM thesis) that the class of recursive algorithms is strictly larger and more expressive than the class of nd-seq algorithms and strictly smaller and less expressive than the class of concurrent algorithms.

As an application of this theory we add to the observation in \cite{gurevich:lipari1995}---namely that recursive algorithms give rise to partial-order runs---a proof that conversely, every finitely composed concurrent algorithm with only partial-order runs is equivalent to a recursive algorithm. This corrects the criticism formulated in \cite{boerger:asm2003} that the answer given by Blass and Gurevich in \cite{blass:beatcs2002} is an ``overkill'', as partial-order runs capture only a restricted concept of concurrency. On the other hand, it underlines also the need for a much more general theory of concurrent algorithms, which is provided by the behavioural theory of concurrent algorithms \cite{boerger:ai2016}.

The debate about ``What is an algorithm?'' is not yet finished, as only an integration of all partial behavioural theories of sequential, recursive, parallel, concurrent, reflective, etc. algorithms\footnote{This list is not yet complete, as in most of the related work mentioned above the aspect of non-determinism as well as randomness is not yet included. However, non-determinism is covered in ASMs \cite{boerger:2003} and is crucial for their applications in system design and analysis.} will provide a final answer to the question. 
A comprehensive definition of the notion of algorithm will require all the known particular classes of algorithms to be integrated in such a way that the specific subclasses arise as special cases of the general definition. However, given that in all the behavioural theories we mentioned above the postulates always concern sequential or branching time, abstract state, background and bounded exploration, the perspective of such an integration looks promising. We invite the reader to contribute to this endeavor.

\bibliographystyle{abbrv}
\bibliography{recursion}

\end{document}